\documentclass[a4paper,10pt,twocolumn,preprint]{elsarticle}
 \usepackage{hyperref}
 \usepackage{amsmath}
 \usepackage{amsthm}
 \usepackage{amssymb}
 \usepackage{stmaryrd}
 \usepackage{algorithm}
 \usepackage{algorithmic}
 \usepackage{multirow}
 \usepackage[usenames,dvipsnames]{xcolor}
 \usepackage{array}

 \newcommand \data[1] {\textsc{#1}}

 \usepackage{paralist}
 \usepackage{caption}
 \usepackage{capt-of}
 \usepackage{subcaption}
 \usepackage[switch]{lineno}

 \newtheorem{proposition}{Proposition}
 \newtheorem{lemma}{Lemma}
 
 \newtheorem{remark}{Remark}

 \DeclareMathOperator{\im}{Im}
 \DeclareMathOperator{\bor}{bd}

 \newcommand{\de}{ \partial }
 \newcommand{\tde}{\tilde{\partial}}
 
 \newcommand{\s}{ \sigma }
 \newcommand{\Ss}{ \Sigma }
 
 \newcommand{\M}{\mathcal{M}}

 \def \Z {{\mathbb Z}}
 \def \R {{\mathbb R}}

\begin{document}

\twocolumn[{
\begin{frontmatter}

\title{Computing discrete Morse complexes from simplicial complexes}
%\tnotetext[mytitlenote]{Fully documented templates are available in the elsarticle package on \href{http://www.ctan.org/tex-archive/macros/latex/contrib/elsarticle}{CTAN}.}

%% or include affiliations in footnotes:
\author[mymainaddress]{Ulderico Fugacci}
\ead[url]{fugacci@.tugraz.at}

\author[mysecondaryaddress]{Federico Iuricich}
\ead[url]{fiurici@clemson.edu}

\author[mythirdaddress]{Leila De Floriani}
\ead[url]{deflo@umiacs.umd.edu}

\address[mymainaddress]{Graz University of Technology, Graz, Austria}
\address[mysecondaryaddress]{Clemson University, Clemson, SC, USA}
\address[mythirdaddress]{University of Maryland, College Park, MD, USA}
% \address[mythirdaddress]{Department of Geographical Sciences, University of Maryland, College Park, MD, USA}

\begin{abstract}
We consider the problem of efficiently computing a discrete Morse complex on simplicial complexes of arbitrary dimension and very large size. Based on a common graph-based formalism, we analyze existing data structures for simplicial complexes, and we define an efficient encoding for the discrete Morse gradient on the most compact of such representations. We theoretically compare methods based on reductions and coreductions for computing a discrete Morse gradient, proving that the combination of reductions and coreductions produces new mutually equivalent approaches. We design and implement a new algorithm for computing a discrete Morse complex on simplicial complexes. We show that our approach scales very well with the size and the dimension of the simplicial complex also through comparisons with the only existing public-domain algorithm for discrete Morse complex computation.  We discuss applications to the computation of multi-parameter persistent homology and of extrema graphs for visualization of time-varying 3D scalar fields.
\end{abstract}


% /footer: © YEAR. Licensed under the CC-BY-NC-ND 4.0 license http://creativecommons.org/licenses/by-nc-nd/4.0/

\end{frontmatter}
}]

% \linenumbers

% !TEX root = Gmod2016.tex
\section{Introduction}
\label{sec:intro}

In recent years, computational topology has become a fundamental tool for the analysis and visualization of scientific data. In particular, the efficient development of software tools for extracting topological features from data has led to an increasing number of applications  of topology-based approaches in shape analysis and understanding,  and in particular in  the analysis of sensor \cite{Silva07} and social \cite{Fellegara2016} networks, in chemistry \cite{martin2010topology}, in astrophysics \cite{van2011alpha}, in medicine \cite{chung2009persistence}. Several mathematical tools have been studied for computing a compact, topologically-equivalent object starting from a simplicial complex of large size. Examples of these tools are the discrete Morse complex \cite{forman1998morse}, the size graph \cite{frosini1999new}, and the tidy set \cite{zomorodian2010tidy}.\\

{\em Discrete Morse Theory (DMT)} \cite{forman1998morse}  is a powerful theory defined in a completely combinatorial setting, that aims at the construction of a discrete representation of  a given  simplicial complex, based on a {\em discrete Morse gradient} (also called {\em Forman gradient} or {\em discrete gradient field}) from which a homology-equivalent chain complex, the {\em discrete Morse complex} is built.
The Forman gradient and the associated discrete Morse complex have been used both for the analysis and visualization of scalar fields \cite{Defl15}, and for computing standard and persistent homology \cite{Robi11, Hark14, harker2010efficiency}.

In very recent research areas, like the analysis of higher dimensional scalar fields \cite{spines} or in the analysis of shapes  based on multi-parameter persistent homology \cite{landi13multi, Allili2017}, there is a need for efficient methods capable of encoding a Forman gradient on higher dimensional simplicial complexes.

In this work, we introduce the first complete study for implementing a Forman gradient on high dimensional simplicial complexes. We start from a theoretical evaluation of the various methods used for building a Forman gradient, which are generally called reduction-based or coreduction-based.
We describe a third method initially formulated in \cite{Fuga14}, obtained by interleaving reductions and coreductions, and we prove the equivalence of all three techniques.
This equivalence will provide us the freedom to implement the method that best fits any given data structure.
% This equivalence provides the freedom to successively adopt the construction method that best fits the most efficient data structure.

To this aim, we undertake a theoretical and experimental evaluation of the three most common data structures for encoding simplicial complexes. Here, we focused on data structures with available public-domain implementations.
Our experiments clearly show that the Generalized Indexed data structure with Adjacencies ($IA^*$) \cite{Cani11}, a data structure encoding only the vertices and a subset of the simplices of the complex, is the only one that can suitably scale to higher dimensions without being affected by the exponential growth in the number of simplices.
We propose a solution to compactly encode a Forman gradient attached to the $IA^*$ data structure.

Based on the latter encoding, we have defined and implemented an efficient, dimension-independent, algorithm for computing a Forman gradient and for retrieving the discrete Morse complex defined by it, which is fundamental for computing, among others, homology and persistent homology. We compare our approach to the one developed in \cite{Hark14} and implemented  in the software library \textit{Perseus} \cite{Perseus} which computes a discrete Morse complex using a data structure
implementing the Hasse diagram of the complex, the {\em Incidence Graph}. Our experiments show that our approach is more efficient and it is also easy to parallelize.

The remainder of the paper is organized as follows. Section \ref{sec:back} introduces some preliminary notions about simplicial complexes, simplicial  and persistent homology, and discrete Morse theory.
Section \ref{sec:works} reviews some classical topological data structures for simplicial complexes as well as  algorithms for computing a Forman gradient and a discrete Morse complex.
In Section \ref{sec:encoding}, we introduce, evaluate and compare the data structures for compactly encoding a simplicial complex and we discuss a new compact encoding for the Forman gradient.
In Section \ref{sec:red&cored&morse}, we present the reduction and coreduction-based algorithms used for computing a Forman gradient. Section \ref{sec:equivalenza} is devoted to the formal proof of the theoretical equivalence of these approaches, while, in Section \ref{sec:metodimisti}, we introduce a new approach based on the interleaving of the two. In Section \ref{sec:algorithm}, we describe a coreduction-based algorithm for building a discrete Morse complex based on the $IA^*$ data structure and on the compact representation of the Forman gradient.
In Section \ref{sec:experimental}, we evaluate the performances of our  algorithm on a variety of input complexes.
%Section \ref{sec:appl} is devoted to discuss how the approach here introduced can immediately lead to some improvements in various application domains.
Finally, in Section \ref{sec:conclusion}, we draw some concluding remarks and discuss applications of our approach to single and multi-parameter
persistent homology computation and to the analysis and visualization of  time-varying 3D scalar fields.

% !TEX root = Gmod2016.tex
\section{Background}
\label{sec:back}

In this section, we introduce some notions which are at the basis of our work. We briefly define and discuss  simplicial complexes, simplicial homology and persistent simplicial homology, as well as discrete Morse theory. %and persistent homology.

\subsection{Simplicial complexes}
\label{subsec:simpl}

A {\em $k$-simplex} $\sigma$ is the convex hull of $k+1$ affinely independent points in the Euclidean space. For instance, a 0-simplex is a single point, a 1-simplex an edge, a 2-simplex a triangle, and a 3-simplex a tetrahedron. Given a $k$-simplex $\sigma$, the {\em dimension} of $\sigma$ is defined to be $k$, and denoted as $dim(\sigma)$. Any simplex $\sigma'$, which is the convex hull of a non-empty subset of the points generating $\sigma$, is called a \textit{face} of $\sigma$.  Conversely, $\sigma$ is called a \textit{coface} of $\sigma'$.

A {\em simplicial complex} $\Ss$ is a finite set of simplices such that:

\begin{itemize}
	\item each face of a simplex in $\Ss$ belongs to $\Ss$;
	\item each non-empty intersection of any two simplices in $\Ss$ is a face of both.
\end{itemize}

We define the \textit{dimension} of a simplicial complex $\Ss$, denoted as $dim(\Ss)$, as the largest dimension of its simplices.
Given a simplex $\sigma$ of $\Sigma$, we define the {\em star} of $\sigma$ as the set of the cofaces of $\sigma$ in $\Sigma$. A simplex $\s$ is called a {\em top simplex} if its star consists only of $\s$ itself. Given a simplex $\sigma$ face/coface of $\sigma'$, $\sigma$ and $\sigma'$ are said to be {\em incident}.  For $k>0$, two $k$-simplices in $\Ss$ are said to be {\em adjacent} if they share a face of dimension $k - 1$, while two 0-simplices $u$ and $v$ in $\Ss$ are called {\em adjacent} if they are both faces of the same 1-simplex.\\\\

Queries on a simplicial complex are often expressed in terms of the topological relations defined by the adjacencies and incidences among its simplices.
% Let $\sigma$ a $k$-simplex of a simplicial complex $\Ss$, we define as
\begin{itemize}
	\item {\em Boundary relations:} given a $q$-simplex $\tau$ and a $k$-simplex $\sigma$ with $q>k$, we say that $\sigma$ is in {\em boundary $(q,k)$-relation} with $\tau$ if $\sigma$ is a face of $\tau$. We denote as  $bd_{q,k}(\tau)$ the set of simplices in boundary $(q,k)$-relation with $\tau$.

	\item {\em Coboundary relations:} given a $q$-simplex $\tau$ and a $k$-simplex $\sigma$ with $q>k$, we say that $\tau$ is in {\em coboundary $(k,q)$-relation} with $\sigma$ if $\tau$ is a coface of $\sigma$. We denote as $cbd_{k,q}(\sigma)$ the set of simplices in coboundary $(k,q)$-relation with $\sigma$.

	\item {\em Adjacency relations:} given two $k$-simplices $\sigma$ and $\sigma'$, we say that $\sigma$ is in  {\em adjacency $(k,k)$-relation} with $\sigma'$ if $\sigma$ is adjacent to $\sigma'$. We denote as $adj_{k,k}(\sigma)$ (or, simply $adj(\sigma)$) the set of simplices in adjacency $(k,k)$-relation with $\sigma$.
\end{itemize}

In the following, we will call {\em immediate} boundary and coboundary relations those boundary and coboundary relations involving simplices of consecutive dimensions.
In the following, we will often refer to them as $bd(\cdot)$ and $cbd(\cdot)$ or, when we need to explicit the complex $\Sigma$ with respect to these relations are considered, as $bd_\Sigma(\cdot)$ and $cbd_\Sigma(\cdot)$.
Figure \ref{fig:relations} illustrates the topological relations of a $1$-simplex (edge) $\sigma_0$ in a simplicial complex $\Ss$. Simplices in the immediate boundary, immediate coboundary and adjacency relations are depicted in blue, red, and green, respectively. Specifically, $bd_{1, 0}(\s_0)=\{v_1, v_3\}$, $cbd_{1, 2}(\s_0)=\{\tau\}$, and $adj_{1, 1}(\s_0)=\{\s_1, \s_2, \s_3, \s_4\}$.

% \begin{itemize}
% 	\item {\em boundary relation}, denoted $R_{k, q}(\s)$, with $k>q$, the set of the $q$-simplices in $\Ss$ which are faces of $\s$;
% 	\item {\em coboundary relation}, denoted $R_{k, q}(\s)$, with $k<q$, the set of the $q$-simplices in $\Ss$ which are cofaces of $\s$;
% 	\item {\em adjacency relation}, denoted $R_{k, k}(\s)$ the set of the $k$-simplices in $\Ss$ adjacent to $\s$.
% \end{itemize}
%In the following, we will call {\em immediate} boundary and coboundary relations those boundary and coboundary relations involving simplices of consecutive dimensions. We will denote such relations as $bd(\cdot)$ and $cbd(\cdot)$, respectively.\\
% \footnote{In any case it could cause an ambiguity, we use the notation $bd_\Ss(\cdot)$ and $cbd_\Ss(\cdot)$ to specify the simplicial complex $\Ss$ with respect to these relations have to be considered.}

\begin{figure}
	\centering
	\includegraphics[width=.45\linewidth]{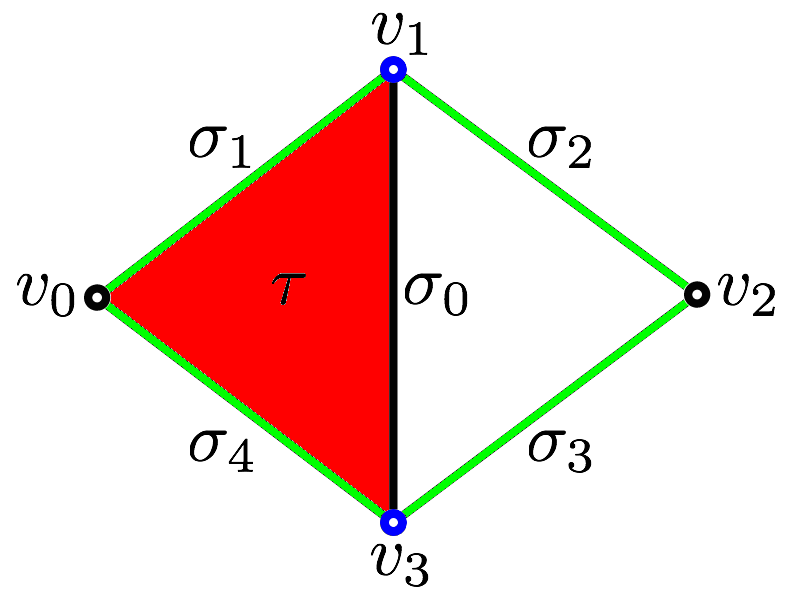}
	\caption{Topological relations of edge $\s_0$. Immediate boundary relation $bd_{1, 0}(\s_0)$ consists of the two blue vertices $v_1$, $v_3$. Immediate coboundary relation $cbd_{1, 2}(\s_0)$ consists of the red triangle $\tau$. Adjacency relation $adj_{1, 1}(\s_0)$ consists of the four green edges $\s_1$, $\s_2$, $\s_3$, $\s_4$.}
	\label{fig:relations}
\end{figure}

Simplicial complexes are a subclass of the more general class of cell complexes \cite{lundell1969topology}. They are extensively used because of their combinatorial properties and of the possibility of representing collections of unorganized sets of points, usually called {\em point clouds}.
{\em Alpha-shapes} \cite{edelsbrunner1983shape}, {\em Delaunay triangulations} \cite{chambers2010vietoris}, {\em \v{C}ech complexes} \cite{hatcher2002algebraic}, {\em Vietoris-Rips complexes} \cite{Zomorodian10fastconstruction}, {\em witness complexes} \cite{DeSilva:2004:TEU:2386332.2386359, Silva03aweak, guibas2008reconstruction} and {\em graph-induced complexes} \cite{Dey:2013:GIC:2462356.2462387} are different ways for endowing a point cloud with a simplicial structure.
{\em \v{C}ech complexes} are the most classical way to build a simplicial complex starting from a point cloud, but their construction requires exponential time in the number of the input points  \cite{Zomorodian10fastconstruction}.

{\em Vietoris-Rips (VR) complexes} \cite{Zomorodian10fastconstruction} represent a compromise between \v{C}ech complexes and the approximations based on subsampling  adopted by witness \cite{DeSilva:2004:TEU:2386332.2386359, Silva03aweak, guibas2008reconstruction} and graph-induced \cite{Dey:2013:GIC:2462356.2462387} complexes. Let $G=(N,A)$ be a graph, a {\em clique} in $G$ is defined as a complete subgraph of $G$. The {\em flag complex} of $G$, denoted as $Flag(G)$, is the simplicial complex whose simplices correspond to the cliques of $G$. Given a finite set of points $P$ in a metric space (such as the Euclidean space) and a positive real number $\epsilon$, the {\em Vietoris-Rips (VR) complex} is the flag complex of the graph whose set of nodes coincides with $P$ and having an arc for each pair of points in $P$ whose distance is at most $\epsilon$. Figure \ref{fig:VR}(a) shows, for each point of a set $P$, the neighboring points at a distance less or equal to $\epsilon$. Figure \ref{fig:VR}(b) shows the edges connecting points in $P$ whose mutual distance is less or equal $\epsilon$. Figure \ref{fig:VR}(c) shows the cliques computed on graph $G$ and the resulting VR complex $\Sigma$.

\begin{figure*}
	\centering
	\begin{tabular}{c c c}
			\includegraphics[width=0.3\linewidth]{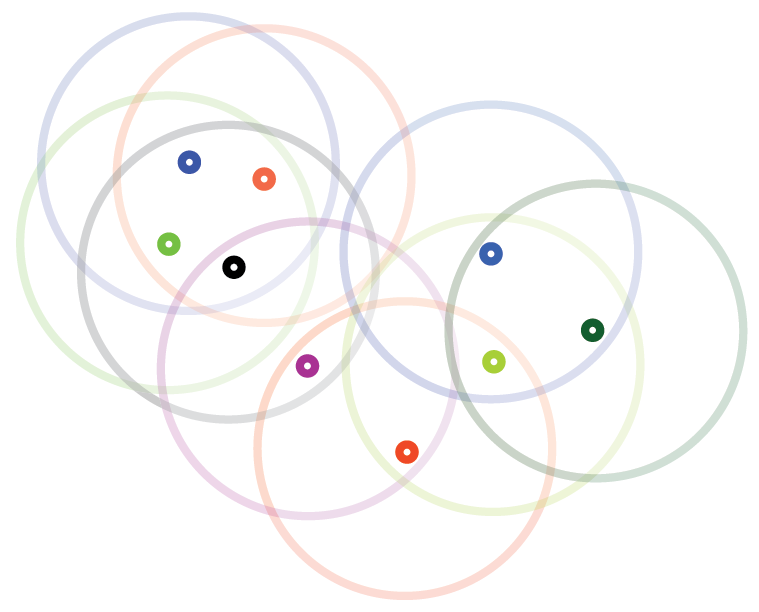} &
			\includegraphics[width=0.3\linewidth]{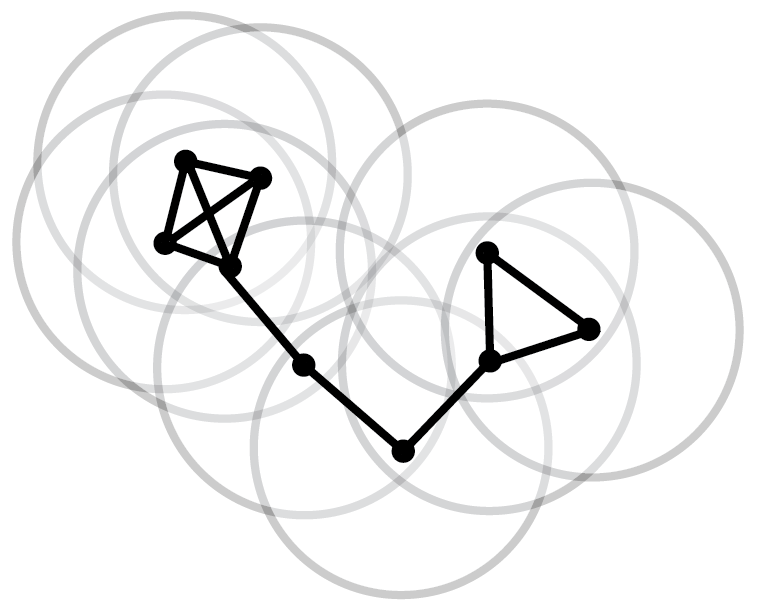} &
			\includegraphics[width=0.3\linewidth]{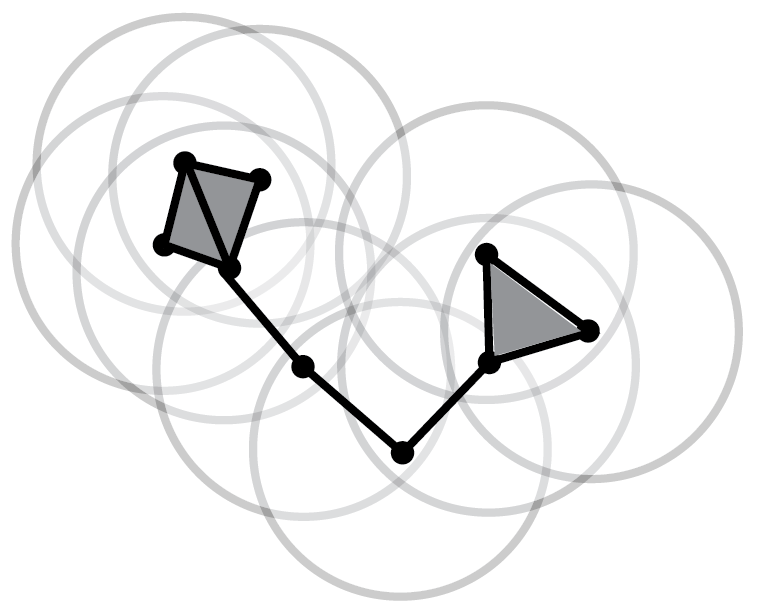}\\
			(a) & (b) & (c)\\
	\end{tabular}
	\caption{Construction of a VR complex $\Sigma$: given a finite set of points $P$, the disks of radius $\epsilon$ are computed (a); an edge is created for each pair of points at distance less than $\epsilon$ (b); the VR complex is retrieved by adding a simplex for each clique of the obtained graph (c).}
	\label{fig:VR}
\end{figure*}

\subsection{Simplicial and persistent homology}
\label{subsec:hom}
{\em Simplicial homology} provides invariants for shape description and characterization.
Given a simplicial complex $\Ss$, we define the {\em chain complex} associated with $\Ss$ as the pair $C_*(\Ss):=(C_k(\Ss), \de_k)_{k\in \Z}$, where:

\begin{itemize}
	\item $C_k(\Ss)$ is the free Abelian group whose elements, called {\em $k$-chains}, are linear combinations with integer coefficients of the $k$-simplices of $\Ss$;
	\item $\de_k:C_k(\Ss) \rightarrow C_{k-1}(\Ss)$ is the homomorphism encoding the boundary relations between the $k$-simplices and those $(k-1)$-simplices of $\Ss$ such that $\de^2=0$.
\end{itemize}

Given $C_*(\Ss)$, we denote as $Z_k(\Ss):=\ker \de_k$ the group of the $k$-cycles of $\Ss$, and as $B_k(\Ss):=\im \de_{k+1}$ the group of the $k$-boundaries of $\Ss$. The \textit{$k^{th}$ homology group} of $\Ss$ is defined as $H_k(\Ss):=H_k(C_*(\Ss))=Z_k(\Ss)/B_k(\Ss)$.
Intuitively, homology groups reveal the presence of ``holes" in a simplicial complex $\Ss$.
The non-null elements of each homology group are cycles, which do not represent the boundary of any collection of simplices of $\Sigma$.
The rank $\beta_k$ of the $k^{th}$ homology group of a simplicial complex $\Ss$ is called the \textit{$k^{th}$ Betti number} of $\Ss$.
In particular, $\beta_0$ counts the number of connected components of $\Ss$, $\beta_1$ its tunnels and holes,  and $\beta_2$ the shells surrounding voids or cavities.\\
%The homological information of a simplicial complex embedded in a space of dimension greater than 3 includes also a {\em torsion part}, which is usually not considered due to inefficiency in computing homology with integer coefficients.\\

\textit{Persistent homology} \cite{edelsbrunner2008persistent, zomorodian2005topology, ghrist2008barcodes}
 aims at overcoming intrinsic limitations of standard homology by allowing for a multi-scale approach defined through a filtration.
Let $\Sigma$ be a simplicial complex, a \textit{filtration $F$} of $\Sigma$ is a finite sequence of subcomplexes
$\{\Sigma^m\, |\, 0\leq m \leq M\}$ of $\Sigma$ such that $\emptyset=\Sigma^0 \subseteq \Sigma^1 \subseteq \dots \subseteq \Sigma^M=\Sigma$. %(see Figure \ref{fig:VR}).
The \textit{$p$-persistent $k^{th}$ homology group $H_k^p(\Sigma^m)$}  of $\Sigma^m$ consists of the $k$-cycles included from $C_k(\Sigma^m)$ into $C_k(\Sigma^{m+p})$ modulo boundaries.

% While homology captures cycles in a complex by factoring out the boundary cycles, persistent homology allows for the retrieval of cycles that are non-boundary elements in a certain step of the filtration and that will turn into boundaries in some subsequent step.
% The persistence of a cycle during the filtration gives quantitative information about the relevance of the cycle itself.

Figure \ref{fig:filtration} shows an example of a filtration of a simplicial complex $\Sigma$. Persistent homology detects the changes in the homology of $\Sigma$ and it allows distinguishing between relevant homology classes, such as the 1-cycle in $\Sigma^1$ which is born at step (a) and persists until the end of the filtration, and negligible homology classes like, for instance, the 1-cycle which is born at step (b) and immediately dies at step (c).

\begin{figure}
	\centering
	\includegraphics[width=\linewidth]{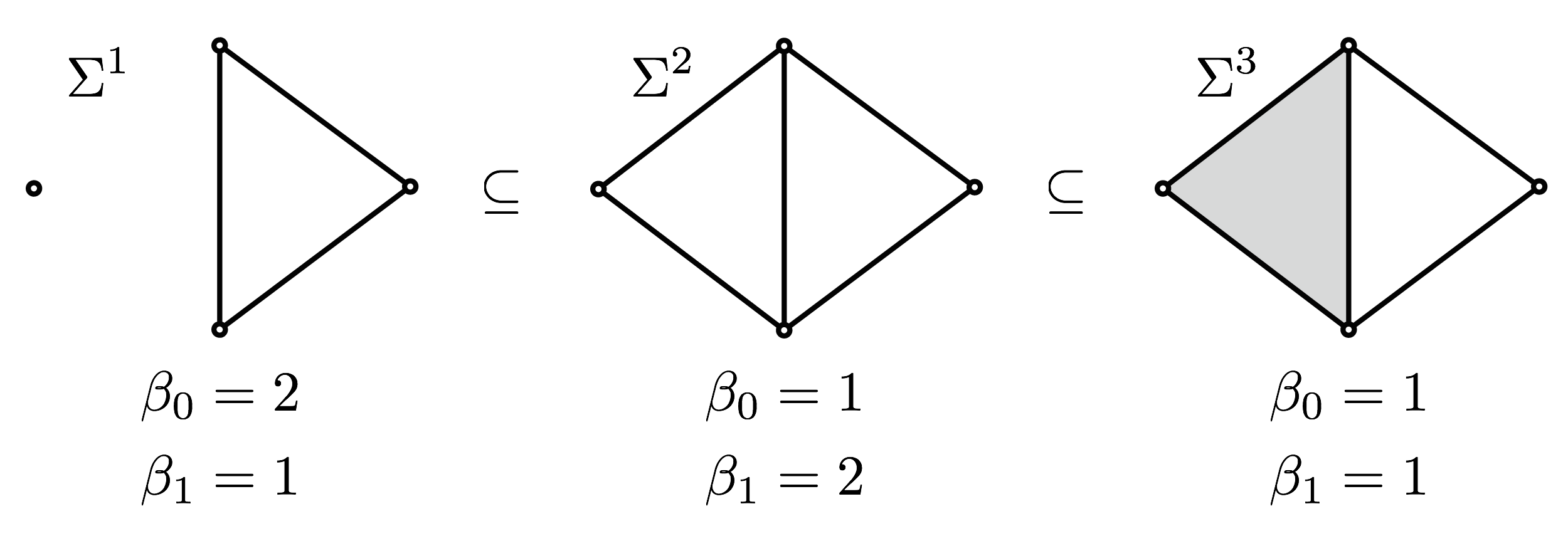}
	\begin{tabular}{c c c}
		 (a) \qquad\qquad\quad\quad &  (b) & \quad\qquad\qquad (c) \\
\end{tabular}
	\caption{A filtration of a simplicial complex $\Sigma$. In (a), $\Sigma^1$ consists of two different connected components and one non-boundary 1-cycle (a); in (b), $\Sigma^2$ gains a non-boundary 1-cycle while it becomes connected; finally, in (c), the 1-cycle created at step (b) becomes the boundary of the unique triangle in $\Sigma^3$ and its contribution in homology vanishes.}
	\label{fig:filtration}
\end{figure}

\subsection{Discrete Morse theory}
\label{subsec:morse}
{\em Discrete Morse theory} due to Forman \cite{forman1998morse, forman2002user} provides a powerful tool for analyzing the topology of an object. It has been defined for cell complexes but, for the sake of simplicity, we will review discrete Morse theory in the context of simplicial complexes.\\
% Given two simplices $\s$, $\tau$ of $\Ss$, we write $\s \prec \tau$ if $\s$ is a face of $\tau$ and $dim(\tau)=dim(\s) +1$.
A simplicial complex $\Ss$  is endowed with a function $f:\Ss\rightarrow \R$, called a {\em discrete Morse function}
if, for every simplex $\s$ in $\Ss$,

\begin{itemize}
\item $c^+(\s):=\#\{ \tau \in cbd(\s) \, | \, f(\tau)\leq f(\s)\} \leq 1$,
\item $c^-(\s):=\#\{ \rho \in bd(\s) \, | \, f(\rho)\geq f(\s)\} \leq 1$.
\end{itemize}

% \begin{itemize}
% \item $c^+(\s):=\#\{\tau \succ \s \, | \, f(\tau)\leq f(\s)\} \leq 1$,
% \item $c^-(\s):=\#\{\rho \prec \s \, | \, f(\rho)\geq f(\s)\} \leq 1$.
% \end{itemize}

It is easy to show (see \cite{forman1998morse}, Lemma 2.5) that, for a discrete Morse function, $c^+(\s)$ and $c^-(\s)$ cannot be simultaneously equal to 1.
A  $k$-simplex $\s$ in $\Ss$ is called {\em critical simplex of index $k$} (or, {\em $k$-saddle}) if $c^+(\s)=c^-(\s)=0$. A critical simplex of index $0$ is called a {\em minimum}, while a critical simplex of index $d=dim(\Ss)$ a {\em maximum}. Figure \ref{fig:forman-gradient}(a) shows a discrete Morse function $f$ defined on a simplicial complex. Each simplex is labeled with the corresponding value of function $f$. Vertex 1 is critical (minimum), since $f$ has a higher value on all edges incident to it. Edge 5 is critical (saddle), since $f$ has a higher value on the incident triangle 7, and lower values on its vertices.

\begin{figure}
	\centering
	\begin{tabular}{c c}
			\includegraphics[width=0.4\linewidth]{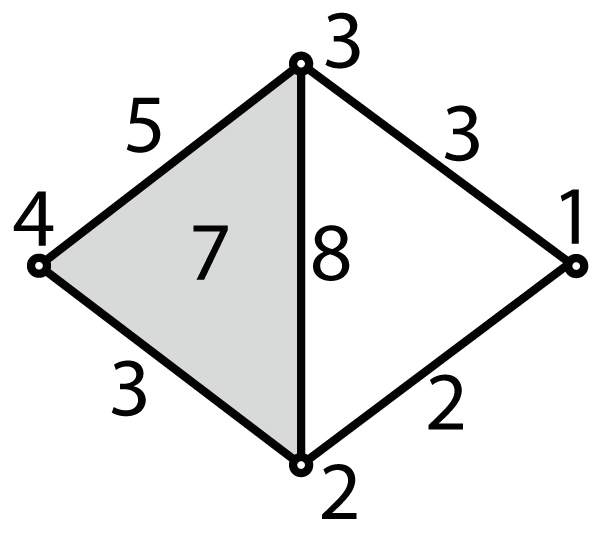} &
			\includegraphics[width=0.4\linewidth]{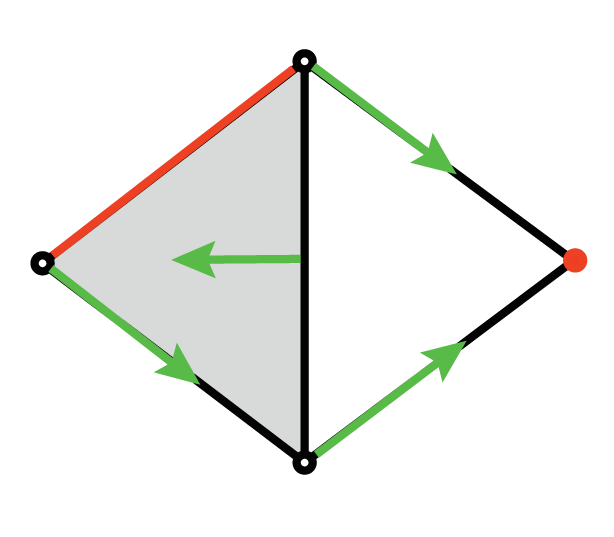}\\
			(a) &  (b)\\
	\end{tabular}
	\caption{(a) A discrete Morse function on a simplicial complex and (b) the corresponding Forman gradient (red simplices are critical simplices).}
	\label{fig:forman-gradient}
\end{figure}

A \textit{discrete vector field $V$} on $\Ss$ is a collection of pairs of simplices $(\s, \tau)\in \Ss \times \Ss$ such that $\s \in bd(\tau)$ and each simplex of $\Ss$ is in at most one pair of $V$.
A discrete Morse function $f:\Ss \rightarrow \R$ induces a discrete vector field
% $V=\{(\s, \tau)\in \Ss \times \Ss\, |\, \s \prec \tau \text{ and } f(\s)\geq f(\tau) \}$
$V=\{(\s, \tau)\in \Ss \times \Ss\, |\, \s \in bd(\tau) \text{ and } f(\s)\geq f(\tau) \}$, called a \textit{Forman gradient} (or, equivalently, \textit{gradient vector field}) of $f$ on $\Ss$.
A pair $(\s, \tau)\in V$  can be depicted as an arrow from $\s$ to $\tau$.
Given a discrete vector field $V$, a {\em $V$-path} (or, equivalently, a {\em gradient path}) is a sequence $[(\sigma_1, \tau_1), (\sigma_2, \tau_2), \dots, (\sigma_{r},\tau_r)]$ of pairs of $k$-simplices $\sigma_i$ and $(k+1)$-simplices $\tau_i$, such that $(\sigma_i, \tau_i)\in V$, $\sigma_{i+1}$ is a face of $\tau_i$, and $\sigma_i\neq\sigma_{i+1}$.
% A $V$-path with $r > 1$ is a {\em closed path} if $\sigma_{1}$ is a face of $\tau_{r}$ different from $\sigma_{r-1}$.
A $V$-path is a {\em closed path} if $\sigma_{1}$ is a face of $\tau_{r}$ different from $\sigma_{r}$.
It has been proven that a discrete vector field $V$ is the Forman gradient of a discrete Morse function if and only if $V$ is free of closed paths \cite{forman1998morse}.

Given a Forman gradient $V$ on a simplicial complex $\Ss$, the {\em discrete Morse complex} associated with $\Ss$ is a chain complex $\M_*:=(\M_k, \tde_k)_{k\in \Z}$, where:

\begin{itemize}
    \item groups $\M_k$ are generated by the critical $k$-simplices;
    \item the boundary maps $\tde_k$ are obtained by following the gradient paths of $V$ (see Subsection \ref{subsec:comp_boundary} for a detailed description).
\end{itemize}

\begin{figure}
	\centering
	\begin{tabular}{c c}
			\includegraphics[width=0.4\linewidth]{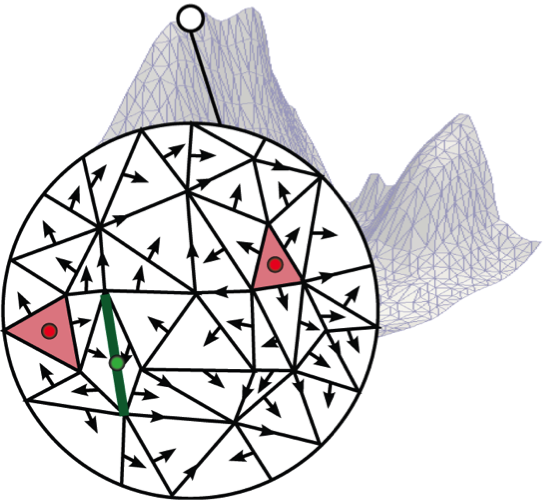} &
			\includegraphics[width=0.5\linewidth]{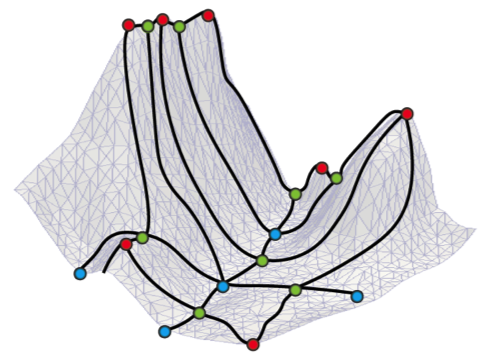}\\
			(a) &  (b)\\
	\end{tabular}
	\caption{(a) A Forman gradient computed on a simplicial complex and (b) the graph structure formed by the gradient paths (boundary maps $\tde_k$) connecting the critical simplices (groups $\M_k$).}
	\label{fig:forman-ig}
\end{figure}

A discrete Morse complex $\M_*$, associated with a simplicial complex $\Ss$, provides a homologically equivalent representation of $\Ss$ (see \cite{forman1998morse}, Theorem. 8.2). If we consider a simplicial complex $\Ss$, and we compute a Forman gradient $V$ on it (see Figure \ref{fig:forman-ig}(a)), we obtain a discrete Morse complex $\M_*$ having its cells in one-to-one correspondence with the critical simplices of $V$. $\M_*$ can be described as a graph having nodes in correspondence of the critical simplices of $V$, and having the arcs in one-to-one correspondence with  the gradient paths connecting such nodes (see Figure \ref{fig:forman-ig}(b)).

Since $\Ss$ and $\M_*$ are homologically equivalent, computing the homology on $\M_*$ is preferable due to the fact that the cells in $\M_*$ are generally fewer than the simplices in $\Ss$. As shown in \cite {mischaikow2013morse}, the homological equivalence between a simplicial complex $\Ss$ and a discrete Morse complex $\M_*$ associated with $\Ss$ can be generalized to persistent homology by requiring that the Forman gradient $V$ is filtered with respect to the filtration $F$ considered.
Formally, given a filtration $F=\{\Ss^m\, |\, 0\leq m \leq M\}$ of a simplicial complex $\Ss$,
a Forman gradient $V$ of $\Ss$ is \textit{filtered} with respect to $F$ if, for each pair
$(\s, \tau)\in V$, there exists $m \in \{1, \dots, M\}$ such that $\s$, $\tau\in \Ss^m$ and $\s$, $\tau\notin \Ss^{m-1}$.
\section{Related work}
\label{sec:works}

In this section, we review the state-of-the-art on data structures for encoding simplicial complexes and on algorithms for computing a discrete Morse complex.

\subsection{Topological data structures for simplicial complexes}
\label{subsec:data_struct}
Several topological data structures for encoding a simplicial complex have been proposed in the literature, mainly for simplicial complexes in low dimensions, and focusing on triangle and tetrahedral meshes (see \cite{DeFl05} for a survey). We consider here data structures specific for simplicial complexes in arbitrary dimensions.

The most general dimension-independent data structure for cell and simplicial complexes is the Incidence Graph.
An {\em Incidence Graph ($IG$)} \cite{Edel87} is a topological incidence-based representation of a simplicial complex which encodes all the simplices as nodes of a graph and their immediate boundary and coboundary relations as its arcs.
% It is an implementation of the Hasse diagram of the complex.
%
% The {\em Simplified Incidence Graph ($SIG$)} \cite{DeFloriani2004data} and the {\em Incidence Simplicial ($IS$)} data structures \cite{DeFloriani2010dimension} are simplified representations of the $IG$ which encode all the simplices plus the same boundary relations as the $IG$, but only a subset of the coboundary relations encoded in the $IG$. Both data structures are more compact than the $IG$ in low dimensions, but they show a similar behavior when working in higher dimensions
% \cite{CanDeF2013imr}.
% %
% $IG$, $SIG$ and $IS$ have been all implemented in the {\em Mangrove Topological Data Structure} library available in the public domain \cite{Mangrove}.
The {\em Simplified Incidence Graph ($SIG$)} \cite{DeFloriani2004data} and the {\em Incidence Simplicial ($IS$)} data structures \cite{DeFloriani2010dimension} are simplified representations of the $IG$. A comparison among $IG$, $SIG$ and $IS$ is presented in \cite{CanDeF2013imr}, while an implementation of all these data structures is included in the {\em Mangrove Topological Data Structure} library available in the public domain \cite{Mangrove}.

% which is a dimension-independent and extensible framework, targeted to the fast prototyping of topological data structures.
% Experimental comparisons of these structures have been performed in 2D, 3D and higher dimensions \cite{Can2012}.
% The $IS$ data structure is as compact as the $SIG$ data structure for 2-dimensional simplicial complexes, and more compact than the $SIG$ representation for 3-dimensional simplicial complexes. Both these structures have shown a behavior similar to the $IG$ when working in high dimensions.
% The $IG$, $SIG$, $IS$ data structures share a common trait storing all the simplices in the simplicial complex.
% In particular, when dealing with data characterized by a huge number of points, or when working in a high-dimensional space, using a data structure which stores all the simplices is unfeasible in practice.
% One could consider data structures encoding only a subset of the simplices.

% Data structures based on encoding only top simplices and vertices have been shown to be particularly effective in two and three dimensions being, in practice, also perfect candidates for extension to higher dimensions \cite{Can2012}.

% In the case of manifold simplicial complexes, adjacency-based data structures have been proven to be more compact, since they encode only vertices and top simplices,  but they are  equally efficient \cite{DeFl05}.
% Thanks to their compactness and efficiency in the case of manifold domains \cite{DeFl05}, adjacency-based data structures are the most widely used for triangle and tetrahedral meshes in a variety of applications.
In the case of triangle and tetrahedral meshes, adjacency-based data structures are the most widely used thanks to their compactness and efficiency.
The {\em Generalized Indexed data structure with Adjacencies ($IA^*$)} \cite{Cani11} generalizes such representations and is capable of encoding non-manifold simplicial complexes of any dimension.
% an adjacency-based representation both  in the dimension of the complex and to non-manifold domains, resulting in the first dimension-independent  data structure for encoding arbitrary simplicial complexes based only on the encoding of vertices and of  top simplices.
% A detailed description of this data structure is provided in Section \ref{sec:encoding}.
Recently, a topological data structure has been proposed for simplicial complexes embedded in the Euclidean space in \cite{fellegara2017stellar}, where topological relations can be efficiently extracted in parallel on different portions of the domain.

% Another data structure based on an explicit encoding only of top simplices and vertices of a simplicial complex is the \emph{Stellar tree} \cite{fellegara2017stellar}. The basis of this data structure is a decomposition of the embedding space of a simplicial complex $\Sigma$ according to a hierarchical spatial index built on its vertices: a block $b$ of the space subdivision contains a subset $V_b$ of the vertices $\Sigma$ plus all top simplices having at least one vertex in $V_b$. This decomposition allows the efficient extraction of other topological relations at runtime \cite{fellegara2017stellar}.\\

% Data structures for simplicial complexes have been developed in recent years, which are geared to perform specific operations efficiently.
In recent years, new data structures have been developed well suited to perform specific tasks.
The  {\em Simplex Tree (ST)}  \cite{Boissonnat2012simplex} has been defined to efficiently extract boundary relations for computing persistent homology. The Simplex Tree encodes all simplices in the complex and tends to be more
verbose than the  $IA^*$ data structure, as shown in  \cite{fellegara2017stellar}. An implementation of  the $ST$  is available in the {\em Gudhi} public domain library \cite{gudhi2014}.
% The {\em Maximal Simplex Tree (MST)}  and the {\em Simplex Array List (SAL)} \cite{Boiss15} have been defined as a compression mechanism of the Simplex Tree. The leaf nodes in the $MST$ are in bijection with the top simplices of $\Ss$, but not all simplices in $\Ss$   are implicitly represented in an $MST$.
% The $SAL$ is a directed acyclic graph storing all the edges of $\Ss$ in the nodes and representing  incidence relations between simplices as arcs, and its space complexity is in between those of an $MST$ and of an $ST$.
The {\em Maximal Simplex Tree (MST)} and the {\em Simplex Array List (SAL)} \cite{Boiss15} are optimized versions of the $ST$. To the extent of our knowledge, there are no implementations of these data structures.
The {\em skeleton blocker} data structure \cite{Attali:2011:EDS:1998196.1998277} has been created specifically to perform edge contraction on a simplicial complex, but it can be efficiently initialized only when working with flag complexes.  An implementation of the latter is provided in the {\em Gudhi} library \cite{gudhi2014}.

\subsection{Computing a discrete Morse complex}
\label{subsec:comp_DMT}

The process of building a discrete Morse complex from a simplicial complex typically consists of two steps: (i) computing the Forman gradient and identifying the critical simplices,  and (ii) extracting the boundary maps. We can classify algorithms for computing a Forman gradient as:  {\em unconstrained} \cite{Lewi03,mrozek2009coreduction, Mrozek:2010:CHA:1872418.1872585,dlotko2011coreduction, Hark14,benedetti2014random} and {\em constrained algorithms} \cite{Caza03,King05,Gyul11-topo-book,Robi11,Gyul12,Gunt12,Vijay12a,Vijay12b}.
\textit{Unconstrained algorithms} compute a Forman gradient on a cell/simplicial complex when no scalar value is provided. The aim is to create a homologically equivalent representation of the input complex having as few critical cells as possible.
\textit{Constrained algorithms} start from a cell/simplicial complex endowed with a scalar function $F_0$ defined on its vertices, and aim at constructing a Forman gradient that best fits $F_0$ \cite{King05, Gyul11-topo-book,Robi11,Gyul12}.
%The applications for constrained algorithms are scientific data visualization and persistent homology computation.
%
The discrete Morse complex is used in the analysis and visualization of scalar fields as a compact representation of the field behavior. The aim is to obtain a decomposition of the dataset in regions of influence for each critical simplex. Ascending and descending traversal techniques \cite{Vijay12b, Felle14, Weis13} for the $V$-paths of the Forman gradient have been developed for reconstructing the ascending and descending Morse cells, respectively.
% These have been defined for triangle meshes \cite{Felle14}, tetrahedral meshes \cite{Weis13} and cubical grids \cite{Vijay12b}.
We refer to \cite{Defl15} for an in-depth description of these methods.
When computing persistent homology on a simplicial complex $\Sigma$ \cite{Zomorodian_computingpersistent,bauer2014phat,distributed}, the aim is to obtain a complex which is a compact version of $\Ss$ and has the same persistent homology \cite{Robi11,mischaikow2013morse, dlotko2014simplification}. To this aim, the gradient V-paths need to be visited by starting from the critical simplices and by traversing the paths in a descending manner. A detailed description of this process is provided in Subsection \ref{subsec:comp_boundary}.

\section{Encoding a simplicial complex endowed with a Forman gradient}
\label{sec:encoding}

In this section, we consider the problem of encoding a simplicial complex endowed with a discrete vector field in a compact way.
%since a fundamental issue when working with a Forman gradient is coupling its encoding to the data structure used for the simplicial complex $\Ss$.
% We start with an analysis of existing data structures for simplicial complexes, that we compare also experimentally, and we then show how we can encode a Forman gradient efficiently using a compact data structure which represents only vertices and  top simplices.
We start with an analysis of existing data structures for simplicial complexes.
Then, driven by the need to identify the most efficient data structure to adopt in our algorithm, we perform an experimental comparison among them. Finally, we show how we can encode a Forman gradient efficiently using a compact data structure which represents only vertices and top simplices.
This is particularly challenging since a representation for a Forman gradient $V$ on a complex $\Ss$ requires encoding the pairings  between two simplices of consecutive dimension for all  simplices in $\Sigma$.

\subsection {Encoding a simplicial complex}
We analyze  here three data structures for encoding a simplicial complex, namely the Incidence Graph ($IG$) \cite{Edel87}, the Simplex Tree ($ST$) \cite{Boissonnat2012simplex,boissonnat2013compressed}, and the Generalized Indexed data structure with Adjacencies ($IA^*$ data structure) \cite{Cani11}. The $IG$ is the most widely-used data structure for simplicial complexes, the $ST$ has been used in topological data analysis applications, being implemented in the {\em Gudhi} library, the $IA^*$ data structure is a compact representation for simplicial complexes encoding only vertices and
top simplices. Implementations in the public domain exist for all of them, and on such implementations we base our experimental comparisons.
%and thus we can perform comparisons among them.

% , for which public domain implementations exist \cite{Mangrove}.

\begin{figure*}
    \centering
    \begin{tabular}{c c c c}
      \includegraphics[width=0.23\linewidth]{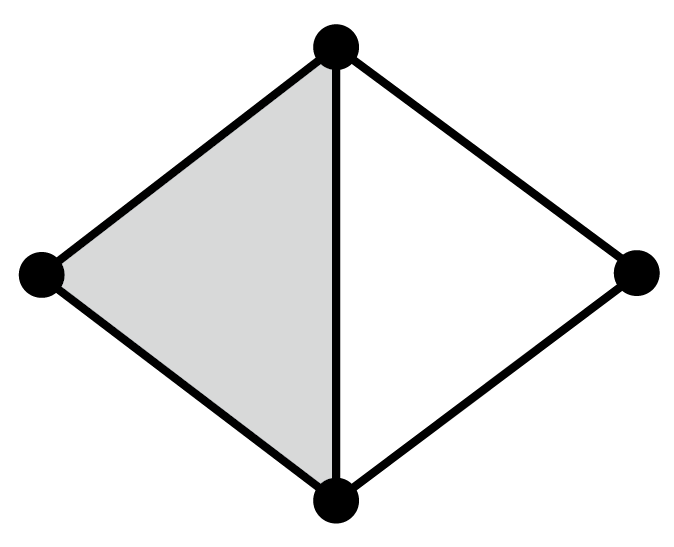} &
            \includegraphics[width=0.23\linewidth]{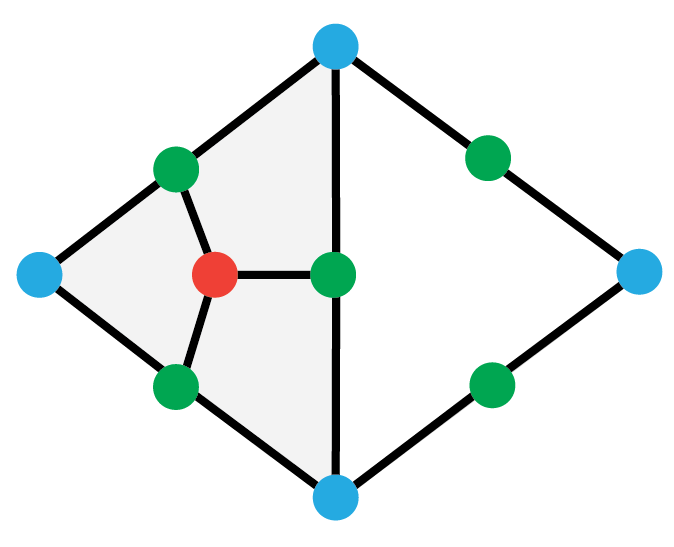} &
            \includegraphics[width=0.23\linewidth]{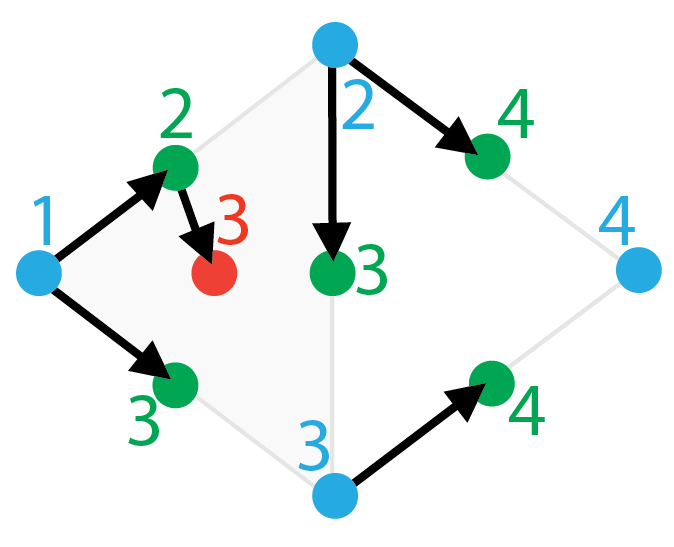} &
            \includegraphics[width=0.23\linewidth]{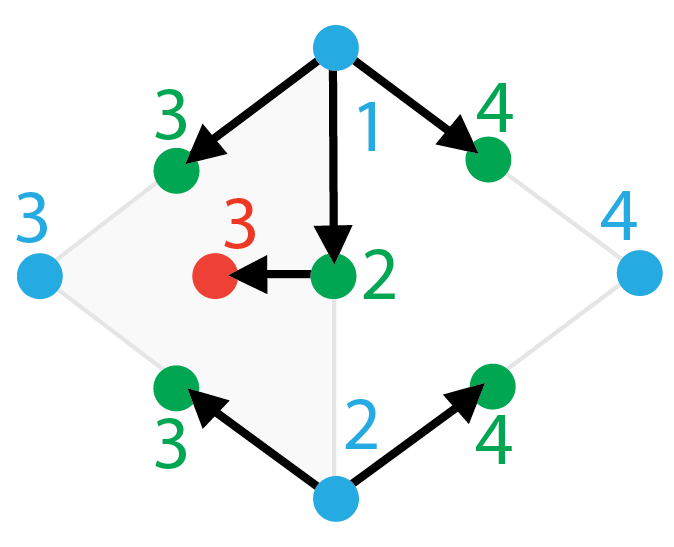} \\
            (a) & (b) & (c) & (d) \\
    \end{tabular}
    \caption{ A simplicial complex $\Ss$ (a) and its representation through an Incidence Graph (b) and  through a Simplex Tree (c); a  Simplex Tree using a different ordering for the vertices (d). Blue dots are associated with the vertices of $\Ss$, green dots with its edges and red dots with its triangles.}
    \label{fig:datastructures1}
\end{figure*}

The {\em Incidence Graph ($IG$)} for complex $\Ss$ describes its Hasse diagram \cite{Pemmaraju}, i.e., the graphical representation of the partially ordered set generated by all the simplices of $\Sigma$ and their incidence relations.
The $IG$ can be viewed as a directed graph $G_{IG}=(N_{IG},B_{IG} \cup C_{IG})$ in which:

\begin{itemize}
  \item the nodes in $N_{IG}$ are in one-to-one correspondence with the simplices of $\Sigma$; with abuse of notation, we will indicate with $\sigma$ both a node, and its corresponding simplex;

  \item a directed arc in $B_{IG}$ ({\em boundary arc}) connects two nodes ($\tau,\sigma$) in $N_{IG}$ with $dim(\tau)=dim(\sigma)+1$ if $\sigma$ is a face of $\tau$;

  \item a directed arc ({\em coboundary arc}) in $C_{IG}$ connects two nodes ($\sigma,\tau$) in $N_{IG}$ with $dim(\tau)=dim(\sigma)+1$ if $\tau$ is a coface of $\sigma$.
\end{itemize}

%Let us consider a partition of the arcs in the $IG$, $A_{IG}=(A_{BD},A_{CB})$, where:
 % \begin{itemize}
  %  \item an arc $(\tau,\sigma) \in A_{BD}$ is a directed arc from $\tau$ to $\sigma$ which represents $\sigma$ as a face of $\tau$;
  % \item an arc $(\sigma,\tau) \in A_{CB}$ is a directed arc from $\sigma$ to $\tau$ which represents $\tau$ as a coface of $\sigma$.
 % \end{itemize}

In Figure \ref{fig:datastructures1}(b), the $IG$ representing the simplicial complex $\Ss$ depicted in Figure \ref{fig:datastructures1}(a) is shown. Nodes are colored according to the dimension of the simplex they represent. Note that, for simplicity, we have shown only one undirected arc for each pair of mutual incident nodes, since if a directed arc exists from $\tau$ to $\sigma$ in $B_{IG}$, arc $(\sigma,\tau)$ must exist in $C_{IG}$
By storing the incidence relations between simplices of consecutive dimension, the $IG$ is efficient in the retrieval of topological relations, but the large amount of information encoded makes it unsuitable for complexes of large size and of high dimensions \cite{CanDeF2013imr}.\\

The {\em Simplex Tree ($ST$)} \cite{Boissonnat2012simplex} encodes also all the simplices of a simplicial complex $\Ss$ as the $IG$, but only a subset of the incidence relations encoded in the $IG$. The $ST$ is based on a total order selected on the vertices of $\Ss$.  Let $I(v)$ be the position in the total order of a vertex $v \in \Ss$. Given a $k$-simplex $\sigma=\{v_0,...,v_k\}$ in $\Ss$,  $max_v(\sigma) = max(I(v_i))$  is the latest vertex of $\sigma$ in the total order. The $ST$ can  be viewed as a graph $G_{ST}=(N_{ST},A_{ST})$ in which:

\begin{itemize}
  \item the nodes in $N_{ST}$ are in one-to-one correspondence with the simplices of $\Sigma$, and a node $\sigma \in N_{ST}$ is labeled with $I(max_v(\sigma))$;
  \item a directed arc $(\sigma,\tau) \in A_{ST}$ connects two nodes in $N_{ST}$, if $\s$ is in the immediate boundary of $\tau$, and $I(max_v(\tau)) > I(max_v(\sigma))$.
\end{itemize}

Nodes corresponding to the vertices of $\Ss$ are connected to the root of the Simplex Tree. If we select a path from the root to a node $\sigma=\{v_0,...,v_k\}$, we have that: (i) labels $\{l_0,...,l_k\}$ are encountered sorted by increasing order along the path and each label appears exactly once;
(ii) each label corresponds to a vertex of $\sigma$, more precisely $l_i=I(v_i)$, for each $i=0,...,k$.

% It encodes all the simplices of the simplicial complex, but only a subset of the boundary/coboundary relations are represented. Chosen a total order among the vertices of $\Ss$, the Simplex Tree encodes the arcs of the Incidence Graph of $\Ss$ which agree with the lexicographic order of the vertices. To given an intuitive definition of the Simplex Tree, we will describe the set of nodes and arcs represented in terms of nodes and arcs of the corresponding $IG$ (see Figure \ref{fig:datastructures1}(b)).
%
% Let $G=(N,A)$ the $IG$ representing a simplicial complex $\Ss$. Let $\{v_0$,...,$v_n\}$ the set of vertices ordered accordingly with a total order, such that vertex $v_i$ has position $i$ in the order.
%
% The Simplex Tree, representing $\Ss$, encodes a node $n_\sigma$ for each node in $N$. We will call dimension of $n_\sigma$, denoted $dim(n_\sigma)$, the dimension of its corresponding simplex $\sigma$, and we will call the label of $n_\sigma$ the position of its latest vertex (i.e., the latest vertex encountered in the total order among its boundary vertices). The Simplex Tree encodes an arc between two nodes $n_\sigma$ and $n_\tau$, with $dim(n_\sigma)=dim(n_\tau)-1$ , if and only if the label of $n_\tau$ is greater than the label of $n_\sigma$.

In Figure \ref{fig:datastructures1}(c), we show the Simplex Tree representation of the simplicial complex depicted in Figure \ref{fig:datastructures1}(a). The order of the vertices is indicated by the numbers depicted in blue, while the remaining numbers indicate the labels of the nodes corresponding to the $k$-simplices, with $k>0$. For the sake of clarity, we are not showing the connections between the vertices and the root. Note that the Simplex Tree is order dependent, in the sense that we can have different $ST$s for the same complex. For example, Figure \ref{fig:datastructures1}(d) shows the $ST$ obtained for $\Ss$ by using a different order for its vertices.
%

%Analyzing the two data structures, we recognize the $IG$ as the most verbose.
From the two graph representations,  we see that $N_{IG} = N_{ST}$ and
that $A_{ST} \subset A_{CB}$, since it contains all those arcs $(\sigma,\tau) \in A_{CB}$ for which $I(max_v(\tau)) > I(max_v(\sigma))$.
The Simplex Tree has been designed with the task of efficiently performing only boundary queries. In order to be able to perform also coboundary queries, an extended version of the Simplex Tree has been proposed in \cite{Boissonnat2012simplex}. This extended version contains a circular list linking all the nodes having the same label and the same dimension and an arc from a node  to its parent.
This version is not implemented in the Simplex Tree in the public domain library {\em Gudhi} \cite{gudhi2014}. In \cite{Boiss15}, two compressed optimization of the latter have been presented, namely the {\em Maximal Simplex Tree} and the {\em Simplex Array List}, sharing the same functionalities but reducing the number of nodes encoded. To the best of our knowledge, no implementations are  provided for these latter.

As mentioned before, more compact representations for a simplicial complex can be obtained by encoding only the vertices and top simplices. To be able to extract boundary, coboundary and adjacency relations efficiently, the simplest representation would encode: (i) for each top $k$-simplex  $\sigma$, its boundary defined by the references to its $k+1$ vertices,  and its adjacencies defined by references to the simplices adjacent to $\sigma$ along a $(k-1)$-face; (ii) for each vertex $v$, its star, defined by the the list of all top simplices incident in $v$. It can be noticed that storing the entire star of a  vertex $v$ is not necessary, since the star can be efficiently reconstructed by navigating the top simplices incident in $v$ through the encoded adjacencies. This constitutes the basis for the  Generalized Indexed data structure with Adjacencies ($IA^*$) \cite{Cani11}.

We can describe the $IA^*$ data structure as a graph $G_{IA}=(N_{IA},A_{IA})$ in which  $N_{IA} = N_0 \cup N_{top}$, with set $N_0$ corresponding to the vertices of $\Sigma$, and set $N_{top}$  corresponding to the top simplices of $\Sigma$. The set of arcs in $A_{IA}$ is the disjoint union of three subsets $A_{(t,0)}$,$A_{(t,t)}$, $A_{(0,t)}$ defined as follows:
\begin{itemize}
  \item $A_{(t,0)}$ ({\em boundary arcs}): a directed arc $(\sigma,v)$, where $\sigma$ is in $N_{top}$ and $v$ in $N_0$, belongs to $A_{(t,0)}$ if  $v$ is a vertex of $\sigma$;
  \item $A_{(t,t)}$ ({\em adjacency arcs}): an undirected  arc $(\sigma,\tau)$, where $\sigma$ and $\tau$ are $k$-simplices in $N_{top}$,  belongs to $A_{(t,t)}$ if $\sigma$ and $\tau$ share a $(k-1)$-face;
  \item $A_{(0,t)}$ ({\em coboundary arcs}): a subset of the arcs $(v,\sigma)$, where $v$ in $N_0$ and  $\sigma$ is in $N_{top}$, such that $v$ is on the coboundary  of $\sigma$, as defined below.
  \end{itemize}

Given a vertex $v$, we consider the subgraph $G_{IA}(v)=(N_{IA}(v),A_{IA}(v))$ of  $G_{IA}$ where:

\begin{itemize}
  \item $N_{IA}(v)$ consists of  all nodes $\sigma \in N_{top}$ such that $v$ is a vertex of $\sigma$;
  \item $A_{IA}(v)$ consists of all arcs in $A_{(t,t)}$ connecting pair of nodes in $N_{IA}(v)$.
\end{itemize}

Thus, an oriented arc $(v,\sigma)$ is encoded in $A_{(0,t)}$ for each connected component in $G_{IA}(v)$, where $\sigma$ is any top simplex  in $N_{IA}(v)$ belonging to such  component.\\

In Figure \ref{fig:datastructures2}, we show the nodes and the arcs encoded in the $IA^*$ data structure (see Figure \ref{fig:datastructures2}(b))  for the simplicial complex in Figure \ref{fig:datastructures2}(a). Blue nodes denote vertices,  while green and red nodes denote top edge and triangles, respectively. Undirected arcs represent adjacency relations among top simplices, i.e., arcs $(\tau_1, \tau_2)$ and $(\sigma_1, \sigma_2)$.
Boundary arcs are denoted by arrows, while coboundary arcs by dotted arrows.
% Directed arcs are used for indicating the incidence relationships from top simplices to vertices. For example, triangle $\tau_1$ is connected to its three vertices $v_0$, $v_1$ and $v_2$. A dotted arc indicates the incidence relation from a vertex to a top simplex. Recall that a node is connected to a number of top simplices equal to the number of components in its star. For example, $v_2$ is connected to $\tau_1$ and to $\sigma_1$.
%The star of $v_2$ is then reconstructed by starting from $\tau_1$ and, navigating by using the  arcs in $A_{(t,t)}$ to $\tau_2$.

%A dotted arc $(\sigma_1,\sigma_2)$ indicate a directional relation, i.e., only $\sigma_1$ encodes the relation with $\sigma_2$.
%Bold arcs indicate relations that are stored in both simplices.
% In the example of Figure \ref{fig:datastructures2} (b), for instance, triangle $\tau_2$ is connected to its three vertices $v_0$, $v_1$ and $v_2$ and  to its 1-adjacent triangles $\tau_1$ and $\tau_3$, and  edge ..is connected to its two vertices .... and to its 0-adjacent edge.... Note that graph $G_{IA}(v_0)$ consists  of  nodes $N = \{\tau_1,\tau_2,\tau_3\}$ and  arcs $(\tau_1, \tau_2)$ and $(\tau_2, \tau_3)$. $G_{IA}(v_0)$ has a single component and, thus, a single arc is stored connecting $v_0$ to $t_2$. The star of $v_0$ is retrieved by starting from $\tau_2$ and, navigating by using the  arcs in  $A_{(k,k)}$ to $\tau_1$ and $\tau_3$.

\begin{figure}
    \centering
    \begin{tabular}{c c}
            \includegraphics[width=0.45\linewidth]{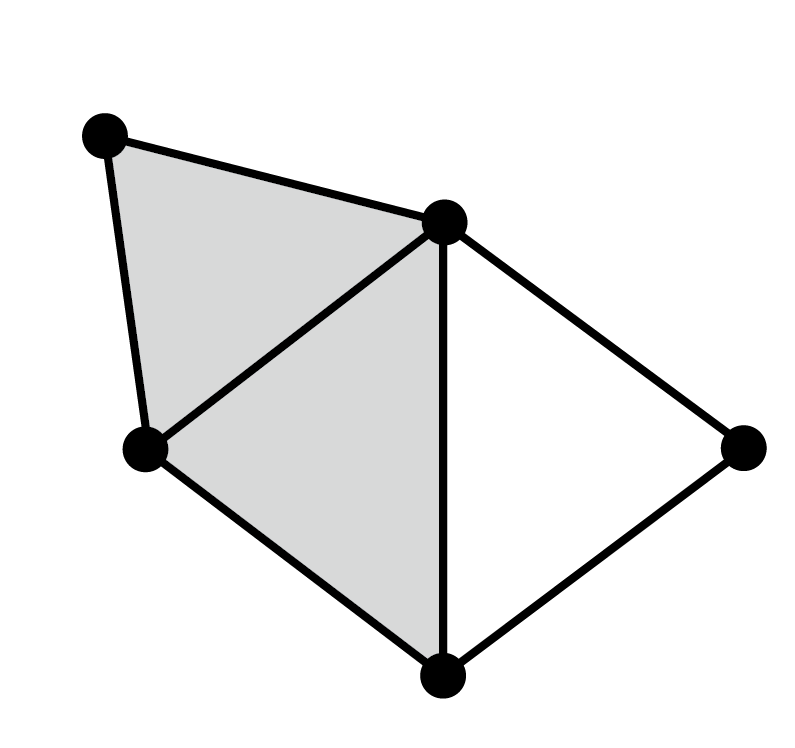} &
            \includegraphics[width=0.45\linewidth]{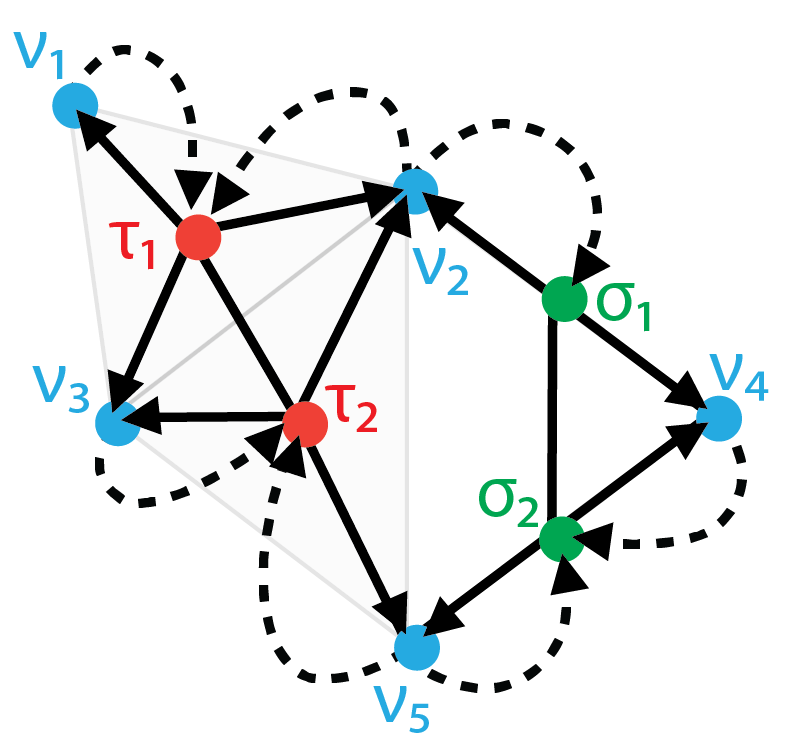} \\
            (a) & (b) \\
    \end{tabular}
    \caption{A simplicial complex $\Ss$ (a) and its representation through the $IA^*$ data structure (b). Blue dots correspond  to vertices, green dots correspond to top edges and red dots to top triangles.}
    \label{fig:datastructures2}
\end{figure}

% The space required by the $IA^*$ data structure is depending on the structure on the complex since it depends on the number of connected components of the stars of its vertices and also on the presence of the non-manifold $h$-simplices which are shared by top $(h+1)$-simplices.

The space required by the $IA^*$ data structure depends on the structure on the complex, i.e., the number of arcs in $A_{(t,t)}$ and in $A_{(0,t)}$ depends on the connectivity of the top simplices.
If we restrict our consideration to an important subclass of simplicial complexes, that of simplicial pseudomanifolds, we can get some insights
for comparing the space required by the $IA^*$ data structure to that of the $IG$. Recall that a {\em simplicial $d$-pseudomanifold} is a $(d-1)$-connected simplicial $d$-complex such that any $(d-1)$-simplex is on the boundary of either one or two $d$-simplices.

If $\Sigma$ is a $d$-pseudomanifold, we have that the number of arcs in $A_{(t,0)}$ originating from a top $k$-simplex $\sigma$ is equal to $k+1$. The number of arcs in $A_{(t,t)}$ originating from a top $k$-simplex $\sigma$ is also equal to $k+1$. Thus, $|A_{(t,t)}|$ is equal to $|A_{(t,0)}|$ and, thus, the total cost of storing the boundary and the adjacency arcs in the $IA^*$ data structure is equal to $2|A_{(t,0)}|$.
We can observe that $c=2|A_{(t,0)}|$ is exactly the cost of storing in the $IG$  all the boundary arcs connecting a $d$-simplex to a $(d-1)$-simplex plus all the dual coboundary arcs connecting a $(d-1)$-simplex to a $d$-simplex. In the $IA^*$ data structure, besides $c$, we have the cost $c_{st}$ of storing some top simplices in the star of the vertices.
For each vertex $v$, $c_{st}$ is equal to the number of connected components of $G_{IA}(v)$. In the worst case this might be equal to the number of $d$-simplices having $v$ on their boundary.
However, in the $IG$ we need to take into account the cost of encoding the other boundary and coboundary arcs which connect $k$- and $(k-1)$-simplices (with $k<d$), which will be clearly much higher than $c_{st}$.

\subsection{Experimental evaluation}
\label{subsec:experimental-evaluation}

%In this section, we evaluate the performances of our coreduction-based algorithm both considering the memory occupied while computing the Forman gradient and the boundary maps, and the timings required. We consider real and synthetic datasets. The hardware configuration used is an Intel i7 3930K CPU at 3.20Ghz with 64GB of RAM.

This subsection provides an experimental comparison among the $IG$, the $ST$ and the $IA^*$ data structure. In our experiments, we have used three kinds of data sets. The first data sets are volume data that have been tetrahedralized. Each vertex of the dataset has an associated scalar value. The \data{DTI-scan} is a Diffusion Tensor MRI Scan of a human brain, the \data{VisMale} dataset is a CT-scan of a man's head and the \data{Ackley} dataset is a synthetic function discretizing Ackley's function \cite{Ackley87}. The datasets in the second group are networks obtained from real data on which cliques have been computed. Two of these datasets (\data{Amazon1}, \data{Amazon2}) are graphs representing the ``Customers Who Bought This Item Also Bought" feature of the Amazon website. If a product $i$ is frequently co-purchased with product $j$, the graph contains a directed edge from $i$ to $j$ (notice, we are considering the graph undirected). The third graph represents a road network in California where intersections and endpoints are described by nodes and the roads connecting these intersections or road endpoints are described by undirected edges (\data{roadnet}). The datasets in the third group are point clouds extracted from a 2-sphere on which a Vietoris-Rips complex has been computed (datasets \data{S1.0}, \data{S1.2}, \data{S1.3}).

%%%LEILA: questa frase da spostare nella parte sperimentale in fondo al lavoro?
%Networks and point cloud datasets have no filtration provided as input.

In our comparisons, we use the  Simplex Tree (ST)  implementation in the {\em Gudhi} library \cite{gudhi2014}, the  Incidence Graph (IG) implemented in {\em Perseus} \cite{Perseus}, which is a public domain tool for computing the discrete Morse complex, and the $IA^*$  data structure implemented in \cite{IAstar}. Table \ref{table:storageStatic} summarizes the characteristics of the datasets we used and their storage costs using the three data structures. For each dataset, we provide the dimension of the resulting simplicial complex (column $d$), the number of its vertices (column $|\Sigma_0|$) and of its top simplices (column $|\Sigma_{top}|$), the size of the complex (column $|\Sigma|$), and the storage cost required by the three data structures, expressed in gigabytes.

 % The {\em Gudhi} library is de-facto the state-of-the-art library for what concern performances and compactness when working on simplicial complexes in high dimensions. Specifically, it has been used for computing persistent homology in combination with the annotation matrices \cite{boissonnat2013compressed}.

\begin{table}
\resizebox{\columnwidth}{!}{%
\begin{tabular}{l | c c c c | c c c}
	\multirow{2}{*}{Dataset} &
	\multirow{2}{*}{$d$} &
	\multirow{2}{*}{$|\Sigma_0|$} &
	\multirow{2}{*}{$|\Sigma_{top}|$} &
	\multirow{2}{*}{$|\Sigma|$} & \multicolumn{3}{|c}{Storage Cost} \\
	& & & & & $IA^*$ & $IG$ & $ST$ \\ \hline

	\data{DTI-scan} & 3 & 0.9M              & 5.5M & 24M  & 0.97  & 11.9   & 2.4 \\
	\data{VisMale}   & 3 & 4.6M              & 26M & 118M & 4.7    &  -      & 9.7 \\
	\data{Ackley4}   & 4 & 1.5M & 32M & 204M & 6.8 & - & 12.8 \\ \hline

	\data{Amazon01}     & 6 & 0.2M       & 0.4M & 2.2M & 0.12 & 1.6 & 0.3 \\
	\data{Amazon02}     & 7 & 0.4M       & 1.0M & 18.4M & 0.28 & 9.8 & 1.5 \\
	\data{Roadnet}         & 3 & 1.9M       & 2.5M & 4.8M & 0.8 & 3.3 & 1.0 \\ \hline

	\data{Sphere-1.0}     & 16 & 100        & 224 & 0.6M & 0.003 & 0.9 & 0.04 \\
	\data{Sphere-1.2}     & 21 & 100        & 285 & 26M & 0.0032 & - & 1.5 \\
	\data{Sphere-1.3}     & 23 & 100        & 382 & 197M & 0.0034 & - & 11.01 \\

\end{tabular}%
}
\caption{Datasets used in the experiments and storage costs for encoding the corresponding simplicial complex with the three data structures $IA^*$, $IG$ and with the $ST$. The storage costs are expressed in gigabytes.}
\label{table:storageStatic}
\end{table}

%Observing the storage cost required for representing $\Sigma$ and $V$, we can notice that the $IG$ requires at least 7 times more memory than the other implementations. For the $IA^*$ implementations, the single thread and parallelized versions require the same amount of memory at runtime. This is an expected result due to the fact that the number of simplices extracted locally to each vertex is small when working in 3D; we will see in the next section that this is not the case when working in higher dimensions. The dedicated $IA^3$ is the most compact data structure since optimized for these kind of datasets and it is generally 2 to 3 times more compact than the generalized version ($IA^*$).\\

We can observe that the storage cost of the $IG$ and of the $ST$ increases based on the total number of simplices. The $IG$ implemented in the {\em Perseus} library often runs out of memory,  while the $ST$ has much higher limits.  The storage cost of the $IA^*$ data structure depends on  the number of top simplices. This means that simplicial complexes in low dimensions (like \data{Roadnet} or the volumetric datasets) may require comparatively more memory than, for example, \data{Sphere-1.3} (being a 23-simplicial complex composed by less than 400 top simplices).
It is clear that the $IA^*$ data structure is always more compact  than the $ST$. The ratio between the storage costs of the two data structures roughly depends on the ratio between the number of top simplices and the size of the complex. The worst-case scenario occurs for (\data{Roadnet}  dataset) where  the  $IA^*$ data structure  requires 20\% less memory than the $ST$, while in  the case of \data{Sphere-1.3}  the storage cost for the $IA^*$ data structure is negligible with respect to the 11 gigabytes required by the $ST$.\\

%%%%LEILA: it would be nice to have a graph showing that he ratio between the two structures roughly depends on the ratio between the number of top simplices and the size of the complex.
%[Ciccio] ci ho provato ma il grafico viene molto brutto. I valori sui primi dataset e quelli sugli ultimi sono a scala molto diversa e questo fa si che si appiattisca tutto

\subsection{Encoding a Forman gradient}
\label{subsec:encoding-gradient}

In this subsection, we describe how to encode a discrete gradient field, like the Forman gradient $V$, on the data structures encoding a simplicial
complex $\Sigma$.

If we consider the Incidence Graph  $G$ representing $\Sigma$,  we see that the arcs of $G$ describe all the possible pairings that can be defined on $\Sigma$ by considering two simplices of consecutive dimension. A Forman gradient can be encoded on the $IG$ by adding one bit flag to each arc $a$  in $C_{IG}$ indicating whether the  nodes incident in $a$ are also a valid pair in $V$. Because of this reason,
%current implementations of the Forman gradient  in low dimensions have been done mainly for cubic grids, in which the $IG$ is implicitly encoded, and
the $IG$ has been selected in the {\em Perseus} tool \cite{Perseus}.
This encoding cannot be extended to the Simplex Tree since this latter encodes only a subset of the coboundary arcs of the $IG$.
% which computes a Forman gradient and the Morse chain complex on such complexes. Moreover, the difficulties in encoding  a discrete gradient field on the Simplex Tree derive from the fact that this latter encodes only a subset of the coboundary arcs of the $IG$, as it is shown in Subsection \ref{sec:encoding}.

%\label{sec:gradient}

%A fundamental issue when working with a Forman gradient is attaching its encoding to the data structure used for the simplicial complex $\Ss$.
%A representation for a Forman gradient $V$ encodes the pairing relation between two simplices, for all the simplices in $\Sigma$.
%
% Given the Incidence Graph $G$ representing $\Sigma$, the arcs of $G$ encode all the possible pairings that can be defined on $\Sigma$ by considering two simplices of consecutive dimension.
%

We describe here a  new representation which allows for a compact encoding of a Forman gradient on the  $IA^*$ data structure and, in general, for any data structure which encodes vertices plus top simplices.
In this case, the encoding for the gradient pairs needs to be attached to the top simplices only. The representation that we have defined encodes, for each top $k$-simplex $\tau$, a bit-vector of length $\sum_{i=1}^{k}{\binom{k+1}{i+1}(i+1) }$ representing all the possible pairings on its boundary.
The first $k+1$ bits encode the pairing between $\tau$ and one of its $(k-1)$-faces. Then, recursively, for each $i$-face of $\tau$, $i+1$ bits are stored until, for each 1-face, 2 bits are encoded storing the pairings with one of its vertices.
For example, considering a 2-simplex (triangle), 3 bits are reserved for encoding the pairings with the boundary edges. Then, for each of them, 2 bits are reserved for encoding the pairings with the boundary vertices (see Figure \ref{fig:pairs}).

\begin{figure}
    \centering
    \includegraphics[width=0.9\linewidth]{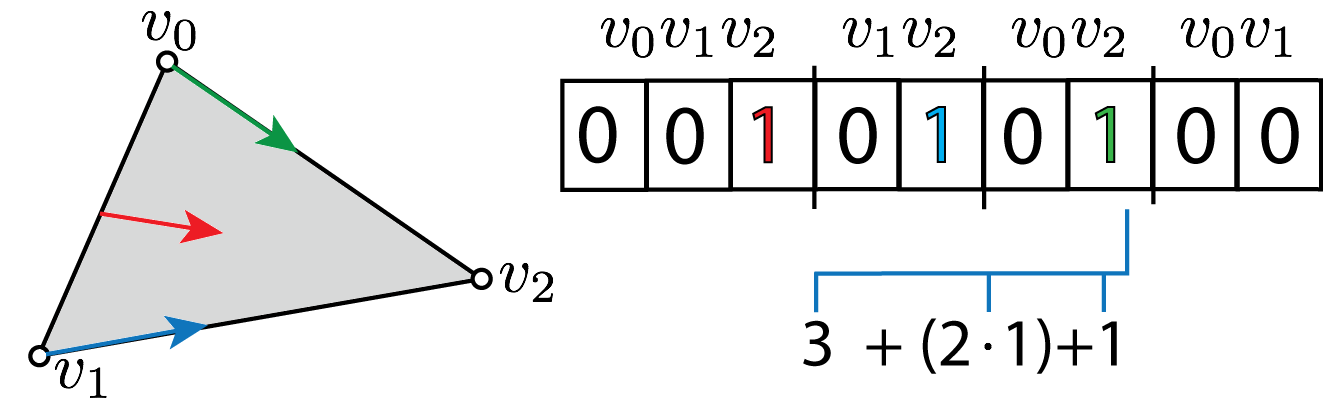}
    \caption{Gradient pairs encoded in a triangle. Pair between vertex $v_0$ and edge $(v_0,v_2)$ is identified by moving on the first bit reserved for the 1-simplices (3 positions). We move forward of one position for each edge preceding $(v_0,v_2)$ on the triangle (2 positions). We do not have to move forward since $v_0$ has position 0 on the edge.}
    \label{fig:pairs}
\end{figure}

If two paired simplices $\rho$ and $\sigma$ are both on the boundary of $\tau$, the resulting pair will be encoded in the bit-vector of $\tau$. Let $j$ and $l$ (with $j+1=l$) be the dimensions of $\rho$ and $\sigma$, respectively, we check the bit associated with the corresponding pair computing:
\begin{itemize}
    \item the position $\sum_{i=l+1}^{k}{\binom{k+1}{i+1} (i+1)}$ of the first bit reserved for $l$-simplices in $\tau$;
    \item the position of $\sigma$ on the boundary of $\tau$ obtained enumerating the faces of $\tau$;
    \item the position of the vertex in $\sigma$ that is not in $\rho$.
\end{itemize}
% \textcolor{red}{Eliminerei $pos_{\rho}$ e $pos_{\sigma}$.}\\
% \nodindent
For example, in Figure \ref{fig:pairs}, we consider the pairing between the 0-simplex $v_0$ and the 1-simplex $v_0v_2$. The bits reserved for the $1$-simplices start at position 3. The position of $v_0v_2$ on the boundary of the triangle is 1, so we discard the first two bits. Vertex $v_2$ is missing in $v_0$ and its position is 1. Then, the bit representing their pairing relation is at position $3 + (2\cdot1) + 1$.

We have implemented a prototype of the gradient encoding based on the {\em dynamic\_bitset} provided by the Boost C++ library. With such encoding, we have been able to represent the gradient frame representation up to 40-dimensional simplicial complexes. Using more involved libraries and architectures could overcome the current limitations, but it might greatly affect computation times.
% (see Section \ref{sec:experimental}).

%For efficiency, we store an additional bit-vector, denoted as  $paired (\sigma)$,  for each top simplex $\sigma$. $paired (\sigma)$ encodes, for each simplex $\tau$ in the boundary of $\sigma$,  whether $\tau$ is paired (or not).  Using such additional bit-vector we are not forced to look outside $\sigma$ for testing if $\tau$ can be paired or not. Notice that, at the end of the coreduction algorithm, all bits corresponding to critical simplices are set as 0 in $paired$.

% !TEX root =  Gmod2016.tex
\section{Reductions and coreductions for  discrete Morse complexes}
\label{sec:red&cored&morse}

\textit{Reduction} and \textit{coreduction} operators  \cite{mrozek2009coreduction} are two homology-preserving operators used for reducing the size of a simplicial complex without affecting its homology.
For this reason, reduction and coreduction operators can be used in a preprocessing approach to compute homology, or persistent homology of a simplicial complex \cite{mischaikow2013morse, mrozek2009coreduction, Mrozek:2010:CHA:1872418.1872585,dlotko2011coreduction}.
Reduction and coreduction pairs can be fruitfully used also in the context of discrete Morse theory in order to define a Forman gradient.
In this section, we present the two methods based on such operators, and we propose a new strategy, while providing also a theoretical comparison of all these techniques.\\

A \textit{reduction} on a simplicial complex $\Ss$ corresponds to a deformation retraction of a simplex which is the face of only one other simplex in the complex. The problem is that, in most situations, available reductions are quickly exhausted.
In order to overcome this issue, coreductions have been introduced \cite{mrozek2009coreduction}, where a coreduction can be viewed as the dual operation with respect to a reduction. A coreduction is not feasible on a simplicial complex, while it is available in the context of \textit{S-complexes} \cite{mrozek2009coreduction}. For the sake of simplicity, we consider an S-complex as a simplicial complex in which some simplices may be not present even if their cofaces are in the complex. For instance, all the complexes depicted in Figure \ref{fig:red-cored} are S-complexes. In particular, the complexes obtained after performing a coreduction operator are examples of S-complexes which are not simplicial complexes.\\
% Let $\Ss$ be an S-complex and let $\s$ be a simplex of $\Ss$. We call the following sets of simplices respectively \textit{(immediate) boundary} and \textit{(immediate) coboundary of $\s$} with respect to $\Ss$
% $$bd_{\Ss}(\s):=\{\rho \in \Ss \, | \, \rho \prec \s \},$$
% $$cbd_{\Ss}(\s):=\{\tau \in \Ss \, | \, \tau \succ \s \}.$$
% A pair $(\s, \tau)$ of elements of $\Ss$, such that the coefficient of $\s$ in $\de\tau$ is $\pm 1$, is  called a \textit{reduction pair} if $cbd_{\Ss}(\s) = \{\tau\}$, a \textit{coreduction pair} if $bd_{\Ss}(\tau)= \{\s\}$.
Given an S-complex $\Ss$, a pair $(\s, \tau)$ of elements of $\Ss$, such that the coefficient of $\s$ in $\de\tau$ is $\pm 1$, is  called a \textit{reduction pair} if $cbd_{\Ss}(\s) = \{\tau\}$, a \textit{coreduction pair} if $bd_{\Ss}(\tau)= \{\s\}$.
% We denote with $bd(\sigma)$ [$cbd(\sigma)$]  the set of simplices in the immediate boundary [coboundary] of $\sigma$.
% Here, we use the same notation adding a subscript for indicating the reference set of simplices. For example, $cbd_{\Sigma}(
% \s)$ indicates the set of simplices in the immediate coboundary of $\sigma$ in  S-complex $\Sigma$.

\begin{figure}
	\centering
	\begin{tabular}{cc}
		\includegraphics[width=.45\linewidth]{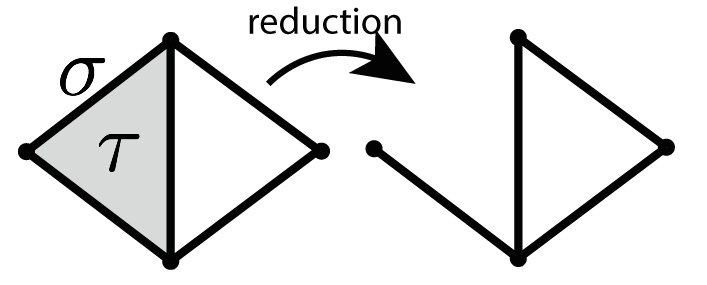} &
		\includegraphics[width=.45\linewidth]{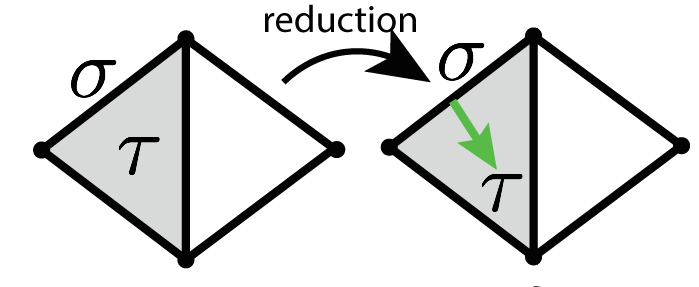} \\

		\includegraphics[width=.45\linewidth]{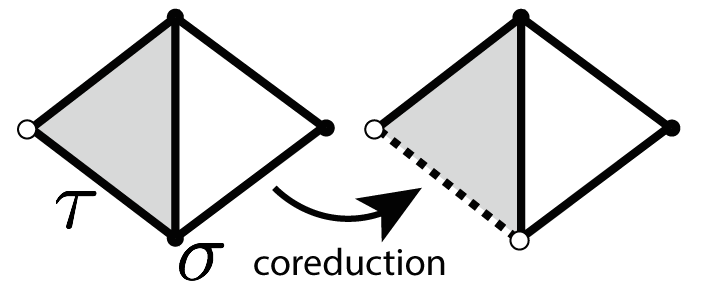} &
		\includegraphics[width=.45\linewidth]{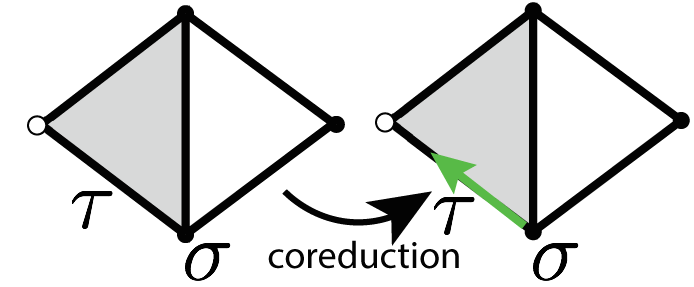} \\

		(a) & (b)\\
	\end{tabular}
	\caption{(a) Removal of the reduction/coreduction pair $(\s, \tau)$, and (b) corresponding pairing of simplices $\s$ and $\tau$ in the gradient.}
	\label{fig:red-cored}
\end{figure}

When simplifying a simplicial complex $\Sigma$, the effect of a reduction/coreduction is that of changing the structure of $\Sigma$, by removing a pair of simplices without affecting its homology (see Figure \ref{fig:red-cored}(a)).
When building a Forman gradient $V$, the same pair is not removed from $\Sigma$, but added as a pair to $V$ (see Figure \ref{fig:red-cored}(b)).\\

A {\em coreduction-based} algorithm builds a Forman gradient using coreduction pairs and free simplices \cite{Hark14}, where a \textit{free simplex} is a simplex with an empty boundary. The algorithm works on two sets of simplices: the set of paired simplices $V$, initialized as empty, and the set of non-excised simplices $\Ss'$, initialized as $\Sigma$.
%\begin{itemize}
 %   \item the set of paired simplices $V$, initialized as empty, and
%    \item the set of non-excised simplices $\Ss'$, initialized as $\Sigma$.
%\end{itemize}
While $\Ss'$ admits a coreduction pair, the algorithm excises a coreduction pair  $(\s, \tau)$  from $\Ss'$ and adds it to $V$. When no more coreduction is feasible, a free simplex is excised from the complex and labeled as critical.
The algorithm repeats these steps until $\Ss'$ is empty. Since no simplicial complex admits a coreduction pair, any coreduction-based algorithm performs as its first step the excision of an arbitrary vertex $v$, which is a free simplex by definition, and declares it as critical. The removal of $v$ turns $\Ss'$ into an S-complex and unlocks the possibility of pairing through a coreduction any vertex $u$ adjacent to $v$.

%%%%LEILA: FATTO. qui bisogna dire che consideriamo all'inizio i vertici e il vuoto. Spiegare che e' un approccio globale quello che viene proposto in letteratura

A {\em reduction-based} approach performs reductions and removals of top simplices \cite{benedetti2014random}. We recall that a \textit{top simplex} is a simplex with an empty coboundary. The algorithm works on two sets of simplices: the set of paired simplices $V$, initialized as empty, and the set of non-excised simplices $\Ss'$, initialized as $\Sigma$.
While the set of non-excised simplices $\Ss'$ admits a reduction pair, the algorithm excises a reduction pair from $\Ss'$ and adds it to $V$.
When no more reduction is feasible, a top simplex is excised from the complex and labeled as critical. The algorithm stops when $\Ss'$ is empty.
%
% In order to minimize the size of the discrete Morse complex, the creation of a critical simplex is performed only if no more reduction pair or coreduction pair is feasible, respectively.
Differently from a coreduction-based algorithm, whose first step is necessarily the removal of a vertex, the initial step in a reduction-based approach can involve the excision of a feasible reduction pair or the removal of a top simplex. Similarly to the previous case, if no reduction pair is available, the approach has to label an arbitrary top simplex as critical and to remove it from $\Ss'$. After such a removal, the situation is analogous to the starting one and, so, the same strategy can be applied.
%%%% LEILA: FATTO. anche qui dobbiamo dire come inizializziamo la ricerca

In order to minimize the size of the discrete Morse complex, in both  approaches the creation of a critical simplex is performed only if no more coreduction, or reduction  is feasible. Actually, even if this condition is not satisfied, the acyclicity of the gradient paths is still guaranteed. In the following, we refer to this two approaches, also in the case in which critical simplices can be created when it is not strictly necessary, as {\em coreduction-based algorithm} and {\em reduction-based algorithm}, respectively.

\section{Equivalence of reduction and coreduction sequences}
\label{sec:equivalenza}

In this section, we prove the equivalence between the use of reduction and coreduction operators in the construction of a (filtered) Forman gradient and we introduce another class of methods which could operate reductions and coreductions in an interleaved way.
% The proposed comparison not only gives an answer to an interesting theoretical question, but it also provides us with a criterion for determining the best strategy to adopt in our algorithm. Indeed, the equivalence we prove in this section will give us the possibility to choose as approach the one that better fits the adopted data structure.
The equivalence among these three methods will give us the freedom to choose the one that best fits our data structure.
% To this aim,  we need some preliminary results which help us to understand how the removal of a coreduction, or of a reduction pair affects the coboundary  and the boundary of the simplices of a simplicial complex.

In order to better understand how the removal of a coreduction, or of a reduction pair affects the coboundary and the boundary of the simplices of a simplicial complex, we first discuss some preliminary results.

\begin{remark}\label{oss1}
Let $\tau$ be a simplex and let $\s$ be one of its faces, then there exists $dim(\tau) - dim(\s)$ faces of $\tau$
in $cbd_\tau(\s)$.
\end{remark}

\begin{lemma}\label{lemma1}
In a coreduction-based algorithm, each removal operation does not modify the coboundary of the remaining simplices.
\end{lemma}

\begin{proof}
	Let $\Ss$ be a simplicial complex on which the coreduction-based algorithm is executed.
	Clearly, the removal of a free simplex does not modify the coboundary of any remaining simplex.
	Let  us consider only removals of coreduction pairs. Let $(\s, \tau)$ be a feasible coreduction pair in the set of non-removed simplices $\Ss'$.
	The only simplices whose coboundary can be modified by the coreduction pair are those belonging to $bd_{\Ss'}(\tau)$ and to $bd_{\Ss'}(\s)$.
	% Since, for the feasible coreduction pair $(\s, \tau)$ $bd_{\Ss'}(\tau)=\{\s\}$,  proving that before  performing  the coreduction $bd_{\Ss'}(\s)=\emptyset$ is sufficient to obtain the thesis.
	Since, for the feasible coreduction pair $(\s, \tau)$, $bd_{\Ss'}(\tau)=\{\s\}$, the thesis is obtained by proving that, before  performing the coreduction, $bd_{\Ss'}(\s)=\emptyset$.
	%%%LEILA: FATTO. non capisco il senso della frase sopra
	Suppose that there exists $\nu \in bd_{\Ss'}(\s)$. By Remark \ref{oss1}, there exists in $\Ss$ a simplex $\s'\neq\s$ such that $\s' \in bd_{\Ss}(\tau)$ and  $\nu\in bd_{\Ss}(\s')$.
	Since $(\s, \tau)$ is a feasible coreduction pair in $\Ss'$, simplex $\s'$ must have been already removed, i.e., $\s' \not\in \Ss'$.
	Let us proceed by induction.
	If $(\s, \tau)$ is the first coreduction pair performed in the coreduction-based algorithm on complex $\Ss$, then $\s'$ has been removed as a free simplex, but, since $\nu\in bd_{\Ss}(\s')$ and $\nu\in\Ss'$, this leads to a contradiction.

Assume now that, for any removal of a coreduction pair performed before $(\s, \tau)$, the simplex of smaller dimension of the pair is free.
Since $\nu\in bd_{\Ss}(\s')$ and $\nu\in\Ss'$, $\s'$ cannot be removed as a free simplex, or by a coreduction pair removal of the kind $(\nu',\s')$.
So, $\s'$ has been removed by operating a coreduction pair removal of the kind $(\s', \tau')$, which leads to a contradiction of the inductive hypothesis.
\end{proof}

\begin{lemma}\label{lemma2}
In a reduction-based algorithm, each removal operation does not modify the boundary of the remaining simplices.
\end{lemma}

\begin{proof}
Let $\Ss$ be a simplicial complex on which the reduction-based algorithm is executed.
Clearly, the removal of a top simplex does not modify the boundary of any remaining simplex.
Let us consider only removals of reduction pairs. Let $(\s, \tau)$ be a feasible reduction pair in the set of non-removed simplices $\Ss'$.
Similarly to Lemma \ref{lemma1}, proving that, before  performing the coreduction, $cbd_{\Ss'}(\tau)=\emptyset$ is sufficient.
If there exists $\nu \in cbd_{\Ss'}(\tau)$, then, by Remark \ref{oss1}, there exist $dim(\nu) - dim(\s) \geq 2$ faces of $\nu$ in $cbd_{\Ss'}(\s)$.
But this leads to a contradiction, because $(\s, \tau)$ is a reduction and, thus, $\#cbd_{\Ss'}(\s)=1$.
\end{proof}

We are now ready to formalize and to prove the equivalence between the coreduction-based and reduction-based algorithms.

\begin{proposition}\label{prop1}
%Coreduction-based and reduction-based algorithms are equivalent. More precisely, g
Given a simplicial complex $\Ss$ and the Forman gradient $V$ produced by a reduction-based algorithm, it is always possible to obtain the same Forman gradient through a coreduction-based algorithm. The reverse is also true.
\end{proposition}

\begin{proof}
For the sake of brevity,  we only prove that the Forman gradient produced by a reduction-based algorithm on $\Ss$ can be obtained with a coreduction-based algorithm. The proof of the reverse is  entirely similar  (by using Lemma \ref{lemma1}).
Let $\Ss$ be a simplicial complex and let
\begin{equation}\label{(2)}
R^1_1, R^1_2,  \dots , R^1_{i_1}, R^2_1, R^2_2,  \dots ,  R^2_{i_2},  \dots ,  R^n_1,  R^n_2,   \dots ,  R^n_{i_n}
\end{equation}
be the ordered sequence of reduction pairs and top simplices removed during the execution of a reduction-based algorithm, where, for $1\leq l\leq n$ and $1\leq j\leq i_l-1$, $R^l_j$ represents a reduction pair and, for each $1\leq l\leq n$, $R^l_{i_l}$ represents a top simplex.

\begin{figure*}
	\centering
	\begin{tabular}{>{\centering}m{0.2in} >{\flushleft\arraybackslash}m{6.5in}}
	(a) & \includegraphics[width=0.88\linewidth]{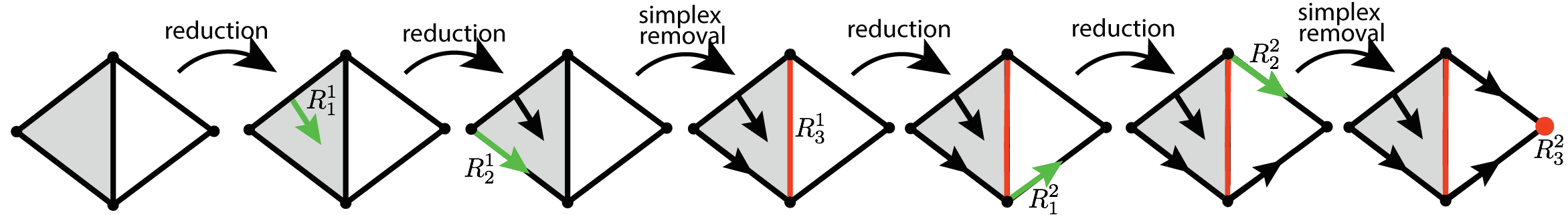} \\
	(b) & \includegraphics[width=0.88\linewidth]{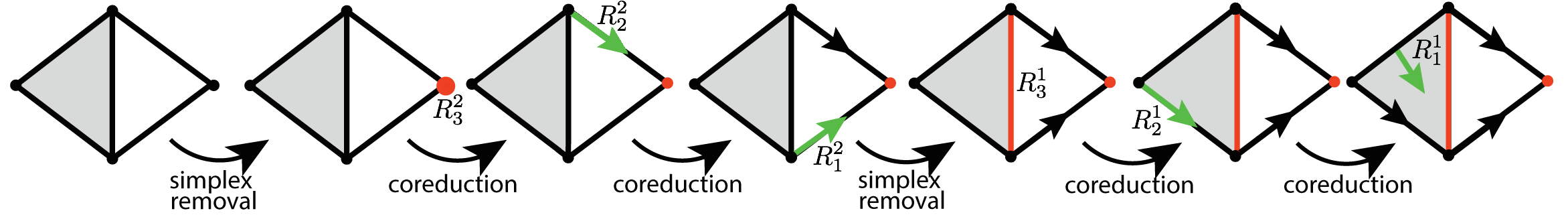} \\
	\end{tabular}
	\caption{(a) A sequence of reduction pairs (green arrows) and top simplex removals (red simplices) produced by a reduction-based algorithm on a simplicial complex and (b) the sequence of coreduction pairs resulting in the same gradient than (a).}
	\label{fig:(3)}
\end{figure*}

According to the notation adopted in \eqref{(2)}, Figure \ref{fig:(3)}(a) depicts the ordered sequence of reduction pairs and top simplices removed during the execution of a reduction-based algorithm. We want to prove that, by using the same removals, it is possible to obtain a sequence of coreduction pairs and free simplices compatible with a coreduction-based algorithm producing the same Forman gradient.
%
% Figure \ref{fig:(3)}(b), for example, shows a sequence of coreduction pairs and free simplices compatible with a coreduction-based algorithm. By reversing the coreduction-based sequence depicted in Figure \ref{fig:(3)}(b), we obtain the reduction-based sequence shown in Figure \ref{fig:(3)}(a) producing the same Forman gradient.
Figure \ref{fig:(3)}(b), for example, shows a sequence of coreduction pairs and free simplices compatible with a coreduction-based algorithm obtained by reversing the reduction-based sequence depicted in Figure \ref{fig:(3)}(a) and producing the same Forman gradient.

We consider the following sequence obtained taking sequence \eqref{(2)} in reverse order:
\begin{equation}\label{(3)}
R^n_{i_n}, R^n_{i_n-1},  \dots , R^n_1, R^{n-1}_{i_{n-1}}, \dots ,  R^1_{i_1},  \dots ,  R^1_2,  R^1_1
\end{equation}

Consider \eqref{(3)} as an ordered list of removal operations performed on $\Ss$. The following properties hold:
\begin{itemize}
\item[1.] For each $1\leq l\leq n$ and $1\leq j\leq i_l-1$, $R^l_j$ is a feasible coreduction pair.
\item[2.] For each $1\leq l\leq n$, $R^l_{i_l}$ is a free simplex.
\end{itemize}

%\begin{figure}
%	\centering
%	\includegraphics[width=1\linewidth]{images/tre}
%	\caption{The sequence of coreduction pairs (blue arrows) and free simplex removals (red simplices) obtained by taking in the reverse order the reduction-based sequence considered in Figure \ref{fig:(2)}.}
%	\label{fig:(3)}
%\end{figure}

To prove the two properties, we denote with:
\begin{itemize}
\item $\Ss^l_j$ the simplicial complex obtained in \eqref{(2)}  after performing all the removal operations up to $R^l_j$ included;
\item $S^l_j$ the S-complex obtained in \eqref{(3)} after performing all the removal operations up to $R^l_j$ excluded.
\end{itemize}
We have that, for each value of $l$ and $j$,
\begin{equation}\label{(4)}
\Ss^l_j\sqcup S^l_j=\Ss
\end{equation}
\\
1. Let $R^l_j=(\s, \tau)$ with $1\leq l\leq n$ and $1\leq j\leq i_l-1$. We have to prove that it represents a coreduction in the sequence \eqref{(3)}, i.e., $bd_{S^l_j}(\tau)=\{\s\}$.
By Lemma \ref{lemma2}, in \eqref{(2)}, $\tau$ cannot be removed before  the simplices in $bd_\Ss(\tau)$. So, all the simplices in $bd_\Ss (\tau)\setminus\{\s\}$ belong to $\Ss^l_j$. Then, by \eqref{(4)}, $bd_{S^l_j}(\tau)=\{\s\}$ and, thus, $(\s, \tau)$ is a feasible coreduction in $S^l_j$.
\\
2. Let $R^l_{i_l}$ be the simplex $\s$. We have to prove that it represents a free simplex in the sequence \eqref{(3)}, i.e., $bd_{S^l_{i_l}}(\s)=\emptyset$.
Analogously to 1., by Lemma \ref{lemma2}, in \eqref{(2)}, all the simplices belonging to $bd_\Ss(\s)$ are in $\Ss^l_{i_l}$. Then, by \eqref{(4)}, $bd_{S^l_{i_l}}(\s)=\emptyset$ and, thus, $\s$ is a free simplex in $S^l_{i_l}$.
\\
Sequence \eqref{(3)} satisfies properties 1. and 2. So, it represents a sequence of removals compatible with a coreduction-based algorithm producing on $\Ss$ the same Forman gradient of \eqref{(2)}.
\end{proof}

It is interesting to understand if the equivalence between reduction-based and coreduction-based algorithms still holds with the further condition that allows for the introduction of a critical simplex only if no reduction [coreduction] pair is available.
Proposition \ref{prop1} ensures that, given a reduction [coreduction] sequence produced on a simplicial complex $\Sigma$ by an algorithm requiring such a condition, it is always possible to find a coreduction [reduction] sequence inducing the Forman gradient on $\Sigma$.
In spite of  this, Proposition \ref{prop1} does not guarantee that a sequence produced by an algorithm satisfying the condition mentioned above exists.
Figure \ref{fig:counterex} shows that, in general, this does not hold.
The Forman gradient depicted in Figure \ref{fig:counterex} can be considered as produced by a reduction-based algorithm starting with the removal of the top simplex $\tau$ and introducing critical simplices only when it is strictly necessary. This Forman gradient cannot be produced by a coreduction-based algorithm in which critical simplices are introduced only when no more coreduction pair is feasible because such an algorithm applied to this simplicial complex necessarily produces a Forman gradient with just one critical simplex of dimension 0 and two critical simplices of dimension 1.

\begin{figure}
	\centering
	\includegraphics[width=.75\linewidth]{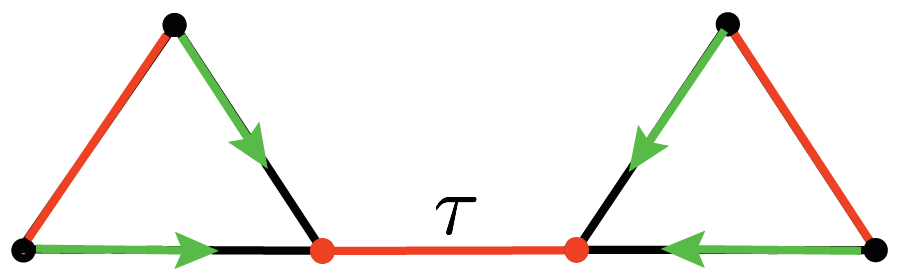}
	\caption{A Forman gradient on a simplicial complex that cannot be produced by a coreduction-based algorithm in which critical simplices are introduced only when no more coreduction pair is feasible.}
	\label{fig:counterex}
\end{figure}

% It is an open problem to find a similar counterexample in the case of a coreduction sequence that has to be interpreted as a reduction sequence or to prove that, in that case, such a situation cannot occur.
% \textcolor{red}{Probabilmente eliminerei quest'ultima frase.}
% E nel caso un tale controesempio non esistesse ci� significherebbe (sicuro?) che, sotto la condizione di minimalita', l'approccio basato sulle coriduzioni risulta essere pi� potente.

\section{Interleaving reductions and coreductions}
\label{sec:metodimisti}

A new method to build a gradient field $V$ on a simplicial complex is to execute removals of reduction and coreduction pairs in an interleaved way.
We denote as \textit{interleaved-based algorithm} an algorithm producing a discrete vector field by using removals of reduction and coreduction pairs, of top simplices and of free simplices. Given a simplicial complex $\Ss$, pairs of simplices are excised from $\Ss$ by arbitrarily choosing between reduction or coreduction pairs. When no more pairs can be removed, a free simplex or a top simplex is excised from the complex and labeled as critical. The algorithm repeats these steps until $\Ss$ is empty.

Here, we prove that such an algorithm actually produces a Forman gradient and that all interleaved methods are equivalent.

\begin{proposition}\label{prop2}
Given a simplicial complex $\Ss$, the discrete vector field $V$ produced by any interleaved-based algorithm is a Forman gradient.
\end{proposition}

\begin{proof}
Given two pairs $(\s, \tau)$, $(\s', \tau')$ in $V$, we define $(\s, \tau) \leq (\s', \tau')$ if there exists a $V$-path starting with $(\s, \tau)$ and ending with $(\s', \tau')$.
In order to prove the thesis, i.e., that $V$ is free of closed $V$-path, it is enough to prove that $\leq$ define a partial order on $V$.
Consider set $V$ as built in any intermediate step of the proposed algorithm and let $(\s, \tau)$ be the last pair inserted in $V$.
% The following properties allow to conclude that the order defined on $V$ is partial:
The following properties allow to achieve the thesis:
\begin{itemize}
\item[1.] $(\s, \tau)$ is a minimal element with respect to the elements already inserted in $V$ originating from a coreduction pair;
\item[2.] $(\s, \tau)$ is a maximal element with respect to the elements already inserted in $V$ originating from a reduction pair.
\end{itemize}
Suppose that condition 1 does not hold. Then, there must exist an already performed coreduction pair $(\s', \tau')$ such that $\s \in bd(\tau')$. This implies that, at the step in which $(\s', \tau')$ has been performed, $\s, \s' \in bd(\tau')$.  But this is impossible, otherwise the coreduction pair $(\s', \tau')$ could not have been performed.
\\
Suppose that condition 2 does not hold. Then, there must exist an already performed reduction pair $(\s', \tau')$ such that $\s' \in bd(\tau)$ and this implies that, at the step in which $(\s', \tau')$ has been performed, $\tau, \tau' \in cbd(\s')$.  But this is impossible, otherwise the reduction pair $(\s', \tau')$ could not have been performed.
\end{proof}

Having proven that any possible interleaved method leads to a Forman gradient, we are now interested in understanding if these different approaches could produce equivalent results or not.
As an immediate consequence of Lemma \ref{lemma1} and Lemma \ref{lemma2}, we can claim the following result.

\begin{remark}\label{oss2}
In each interleaved-based algorithm,
each coreduction pair and free simplex removal cannot make a reduction pair feasible;
each reduction pair and top simplex removal cannot make a coreduction pair feasible.
\end{remark}

Finally, we can prove that all interleaved methods are equivalent.

\begin{proposition}\label{prop3}
Given a simplicial complex $\Ss$ and the Forman gradient $V$ on it produced by an interleaved-based algorithm, it is always possible to obtain the same Forman gradient with a reduction-based algorithm or, equivalently, with a coreduction-based algorithm.
\end{proposition}

\begin{proof}
We prove that the sequence of removals produced by an interleaved-based algorithm on a simplicial complex can be also obtained with a sequence of coreduction pairs and free simplex removals.
By Remark \ref{oss2}, we can suitably order such a sequence, moving all the coreduction pairs and the free simplices at the beginning,  thus creating a new sequence equivalent to the previous one.
We apply to the last part, composed only of reduction pairs and top simplices, of this new sequence the same sorting strategy proposed in Proposition \ref{prop1} to transform a reduction-based sequence to a coreduction-based sequence, and in this way, we obtain the thesis.
\end{proof}

From both an application and a theoretical point of view, it is interesting to find a method to build a Forman gradient which minimizes the number of critical simplices. It is known that, in general, this problem is NP-hard \cite{NPhard}. The previous results show that, from a theoretical point of view, the use of different simplification operators (such as reduction and coreduction pairs), or the combination of more than one, does not actually affect the number of  resulting critical simplices.

% An analogous conclusion is reached also for the construction of a filtered Forman gradient.
% Let $\Ss$ be a simplicial complex, $F$ a filtration of $\Ss$ and $V$ a Forman gradient on $\Ss$.
% For each pair $(\s, \tau)\in V$, the satisfaction of the condition required to guarantee that $V$ is a filtered Forman gradient of $F$ is independent of the fact that $(\s, \tau)$ has been created thanks to a reduction or a coreduction pair. So, the results proved in this and in the previous section still hold when the above-described approaches are applied to build a filtered Forman gradient useful for the persistent homology computation.

For the sake of completeness, let us note  that the results proven in this  section still hold when the above-described approaches are applied to build a filtered Forman gradient. This is due to the fact that the satisfaction of the condition required to guarantee that $V$ is a filtered Forman gradient with respect to a filtration $F$ does not take into account if the pairs of $V$ have been created thanks to a reduction or a coreduction operator.

% !TEX root = Gmod2016.tex
\section{A coreduction-based algorithm for computing a discrete Morse complex}
\label{sec:algorithm}

In this section, we describe an algorithm based on the $IA^*$ data structure for computing a discrete Morse complex. The algorithm consists of two steps: (i) computation of a (filtered) Forman gradient through a coreduction-based approach, and (ii) extraction of the boundary maps defining the discrete Morse complex.

\subsection{Construction of a (filtered) Forman gradient}
\label{subsec:comp_gradient}

The theoretical equivalences proven in Section \ref{sec:equivalenza} and Section \ref{sec:metodimisti} tell us that there is no preferable homology-preserving operator for computing a Forman gradient. Here, we introduce a new dimension-independent algorithm, that can also runs in parallel, which uses a representation of the simplicial complex as an  $IA^*$ data structure and the encoding of the Forman gradient discussed in Subsection \ref{sec:encoding}.

The basic underlying approach is the  coreduction-based algorithm, introduced in \cite{Hark14} and implemented there only for regular grids.
%, that we have described in Section \ref{sec:red&cored&morse} for simplicial complexes.
We summarize it for simplicial complexes.
When considering simplicial complexes, the coreduction-based algorithm computes a Forman gradient by using coreduction pairs starting from the simplices of lowest dimension. The set of $k$-simplices of complex $\Sigma$ is considered by increasing values of $k$, starting from $k=0$. As long as a coreduction pair exists between a $k$-simplex $\s$ and a $(k+1)$-simplex $\tau$, pair $(\s,\tau)$ is added to the Forman gradient $V$. When no $k$-simplex can be paired, one simplex is randomly chosen and declared critical. When all $k$-simplices have been paired or denoted as critical, the working dimension $k$ is increased by one. Since no coreduction pair is feasible on a simplicial complex, at the first step, an arbitrary vertex $v$ is denoted as critical in $V$ to trigger coreductions. In \cite{mischaikow2013morse},  the coreduction-based approach is used for persistent homology computation, and thus by considering a filtration of the original complex. If each simplex is paired only with another simplex belonging to the same filtration value,  the resulting discrete Morse complex will have the same persistent homology of the original complex.

The dimension-independent coreduction-based algorithm proposed here, unlike previous ones, uses a local approach that allows us to work on the stars of the vertices independently, which makes it particularly suitable for a parallel implementation. We define an indexing on the vertices of the input simplicial complex  $\Ss$, and we extend the indexing to all the simplices in such a way that each simplex in $\Ss$ has an index equal to the maximum of the indexes of its vertices.  With such indexing,  the coreduction pairs can be computed locally to the lower star of each vertex. Given a vertex $v$, a simplex $\s$ belongs to the {\em lower star} of $v$ (denoted  as $St^-(v)$) if: (i) $\s$ is a coface of $v$, and (ii) $v$ has lowest index value among the vertices of $\s$. {\em Algorithm} \ref{alg:coreduction} illustrates the process for computing a Forman gradient on a simplicial complex $\Sigma$ having an indexing $F_0$ defined on its vertices.

\begin{algorithm}[!htb]\caption{ - \textbf{FormanGradient($\Sigma$,$F_0$)}}
    \label{alg:coreduction}
    \begin{algorithmic}[1]
        \medskip
        \STATE{INPUT: $\Sigma$, $d$-dimensional simplicial complex}
	\STATE{INPUT: $F_0$, indexing of vertices of $\Sigma$}
        \STATE{OUTPUT: $V$, Forman gradient; $C$, set of critical simplices}
        \STATE{$\Sigma_0$ := vertices of $\Sigma$}
        \STATE{$V :=\emptyset$}
        \STATE{$C :=\emptyset$}
        \FOR{$v \in \Sigma_0$}
        		\STATE{$k := 0$}
      % \STATE{ \COMMENT{compute the $k$-simplices in the lower star} }
      % \STATE{$ST_v := LowerStar(v,\Sigma, F_0 , k)$}
      \STATE{$ST_v := \{v\}$}
      \STATE{$LT_v := LowerTop(v,\Sigma,F_0)$}

				\WHILE{$k <= d$}
					\STATE{$CR_v := ST_v$}
          % \COMMENT{save unpaired $k$-simplices}
					\STATE{$k := k+1$}
          % \STATE{\COMMENT{compute $k$-simplices of the lower star}}
					\STATE{$ST_v := LowerStar(v,\Sigma, F_0, LT_v, k)$}
					\WHILE{$CR_v \,\neq \emptyset$}
						\STATE{$(\sigma,\tau) := getNextPair(v, \Ss, ST_v, CR_v)$}
						\IF{$(\sigma,\tau)  \,\neq \emptyset$}
							\STATE{$addPair(\sigma,\tau,V)$}
              % \COMMENT{new pair}
							\STATE{$Remove(\tau, ST_v)$}
							\STATE{$Remove(\sigma, CR_v)$}
						\ELSE
							\STATE{$\sigma = getFirstCritical(CR_v)$}
              % \STATE{\COMMENT{new critical simplex}}
							\STATE{$addCritical(\sigma, C)$}
							\STATE{$remove(\sigma, CR_v)$}
						\ENDIF
					\ENDWHILE
				\ENDWHILE
			%\ENDIF

        \ENDFOR

    \end{algorithmic}
\end{algorithm}

The algorithm iterates on the vertices of $\Ss$, extracting first the top simplices in the lower star of $v$,  denoted as $LT_v$,  which are encoded in a list.
For each vertex $v$, the algorithm iterates on the dimension of the simplices in the lower star of $v$.
The algorithm works, for each dimension, with two sets of simplices: the set of $k$-simplices that can be declared critical, denoted as $CR_v$ (row 12),  and the set of $(k+1)$-simplices to pair (row 14), denoted as $ST_v$. $CR_v$ and $ST_v$ have a maximum size equal to the maximum, by varying $k$, of the number of $k$-simplices in the lower star of a vertex in $\Ss$, and they are encoded as balanced binary search trees.
A candidate simplex is extracted from set $ST_v$ (row 16) and paired with its unique unpaired face (row 18). Recall that a simplex $\tau$ can be paired with another simplex $\sigma$ by coreduction if $\sigma$ is the only unpaired face of $\tau$. If there are no coreductions available (row 21), a new critical simplex is taken from $CR_v$. Every time a simplex is paired or set as critical, it is also removed from $ST_v$ or $CR_v$, respectively. When set $CR_v$ is empty, the working dimension is increased. The algorithm terminates when all the simplices in the lower star of  each vertex  $v$ have been paired, or set as critical.

 %is achieved by using the {\em sets} implemented
%as binary search trees in the std C++ library.
%The space to be allocated for storing any of these two sets is bounded by the maximum, by varying $k$, of the number of $k$-simplices in the lower star of $v$.
%In order to quickly retrieve the $k$-simplices belonging to the lower star of a vertex $v$, the list $LT_v$ of the top simplices of $\Sigma$ contained in $St^-(v)$ is encoded as an auxiliary structure. In the worst case, the maximum size of this list coincides with the cardinality of $St^-(v)$.
%Algorithm \ref{alg:coreduction} is based on the following functions and procedures.
The procedures and the functions, on which Algorithm \ref{alg:coreduction} is based, are:
\begin{itemize}
		\item Function $LowerTop(v,\Sigma,F_0)$: computes all the top simplices of $\Sigma$ belonging to the lower star of $v$  and encodes such simplices in list $LT_v$. This is performed by navigating the star of vertex $v$ through the adjacency arcs in the $IA^*$ data structure. Thus, it works in time $O(t_v)$, where $t_v$ denotes the number of top simplices in the star of $v$.

% top simplices incident in $v$ can be retrieved in linear time $O(t_v)$ with respect to their total number $t_v$.

    \item Function $LowerStar(v,\Sigma,F_0,LT_v,k)$: extracts all the $k$-simplices belonging to the lower star of $v$ from $LT_v$ and encodes such simplices in $ST_v$. This operation is performed by cycling on the elements of $LT_v$ and collecting the $k$-faces of each top simplex that are also incident in $v$. The  extraction of the $k$-simplices of a top simplex of dimension $i$ is performed in $O({{i+1}\choose{k+1}})$. If we denote as $t_{v,i}$ the number of top simplices of dimension $i$ incident in $v$, the total number $N_k$ of simplices extracted is $N_k=\sum_{i=1}^{d}t_{v,i}{{i+1}\choose{k+1}})$ in the worst case, since some simplices are contained within the boundary of more than one top simplex. Since each of such simplices is inserted in $ST_v$,   $LowerStar(v,\Sigma, F_0,LT_v,k)$ may require $O(N_k \log N_k)$ time in the worst case.
		% Let $t_v$ denote the number of top simplices incident in $v$ and $s_v$ denote the number of $k$-simplices in the lower star of $v$. The procedure $LowerStar(v,\Sigma, F_0, LT_v, k)$ takes $O(t_v  s_v)$ time in the worst case since some simplices are contained within the boundary of more than one top simplex.

    \item Procedure $addPair(\sigma, \tau, V)$:  adds a new pair to $V$. Since the gradient pairs are encoded on the top simplices only, we have to
    find the top simplices incident in both $\sigma$ and $\tau$. This is done by examining all the top simplices in the star of a vertex of $\sigma$ and detecting all
    those having $\tau$ on their boundary. For each of these latter, we update the corresponding bit-vector. The operation requires $O(t_{w})$, where $t_w$ denotes the number of top simplices incident in a vertex $w$ of $\sigma$.

    \item Function $getNextPair(v, \Ss, ST_v, CR_v)$: iterates on the set of unpaired simplices $ST_v$ selecting the first simplex available for a coreduction. For each simplex $\tau$ in $ST_v$, the simplices on the boundary of $\tau$ containing $v$ are extracted, and then, for each of such boundary simplices $\sigma$, the membership of $\sigma$ to $CR_v$ is checked. In the worst case, we will need to check the membership of all the elements in $CR_v$. This leads to a worst-case complexity of $O(k|ST_v||CR_v|\log |CR_v|)$.

    \item Function $getFirstCritical(CR_v)$: returns the first simplex in the set of candidate critical simplices $CR_v$. Since $CR_v$ is implemented as a balanced binary search tree, the worst-case time complexity is $O(\log |CR_v|)$.

    \item Procedure $Remove(\s,CR_v)$: eliminates a simplex from $CR_v$ (or $ST_v$). Since both $CR_v$  and $ST_v$  are  implemented as a balanced binary search tree, the worst-case time complexity is $O(\log |CR_v|)$.

\end{itemize}

For each dimension, the computation cost is dominated by the cost of executing Function $getNextPair(ST_v, CR_v, \Ss)$. If we denote as $cr_m$
and as $st_m$ the maximum size of $CR_v$ and $ST_v$, respectively over all dimensions, the time complexity for a single vertex $v$ is  $O((d-1)st_m cr_m \log (cr_m))$. Note that both $cr_m$ and $st_m$  can be of the order of the number of $k$-simplices incident in $v$.
%For each vertex $v \in \Sigma_0$, all the simplices of a given dimension $k$ are extracted from its lower star.
%For each dimension, the cost of choosing a simplex to pair is $O(|ST_v|)$ (due to $getNextPair(ST_v, CR_v, \Ss)$). This is repeated until $CR_v$ is empty, thus leading to a worst case complexity of $O(|ST_v| \cdot |CR_v|)$.
%
Since the algorithm computes the Forman gradient locally to the lower star of each vertex, the approach is easy to parallelize by running Algorithm \ref{alg:coreduction} on multiple vertices at a time. Results are shown in Section \ref{sec:experimental}.

We prove the correctness of Algorithm \ref{alg:coreduction}  by showing that it is a coreduction-based algorithm ensuring that the generated discrete vector field $V$ is a filtered Forman gradient.

\begin{proposition}\label{pro:cored} Let $\Ss$ be a simplicial complex, $F_0 : \Ss_0 \rightarrow \R$ be an injective function and $F$ be the filtration of $\Ss$ naturally induced by $F_0$. Given $\Ss$ and $F_0$ as input, Algorithm \ref{alg:coreduction} returns a filtered Forman gradient with respect to $F$.
\end{proposition}

\begin{proof} Algorithm \ref{alg:coreduction} processes the lower stars of the vertices of $\Ss$ independently. Without loss of generality, we can assume that the lower stars are processed in a sequence ordered by ascending values of function $F_0$. In this way, we obtain an ordered sequence of simplices added to the gradient $V$ and to the set of critical simplices $C$. We prove that this sequence, denoted as $S$, actually represents a feasible sequence of coreduction pairs and free simplices for $\Ss$. Let us consider a pair of simplices $(\s, \tau)$ declared as a pair of $V$ during the processing of the lower star $St^-(v)$ of $v$. Let $\s'$ be a simplex in $\bor_\Ss \tau$ different from $\s$. If $\s'\in St^-(v)$, then $\s'$ has to be already added to $V$ or to $C$. Otherwise, if $\s'\not\in  St^-(v)$, then there exists a vertex $w$ of $\Ss$ such that $\s'\in St^-(w)$ and $F_0(w) < F_0(v)$. So, $\s'$ has to be already added to $V$ or to $C$ during the processing of $St^-(w)$. In both cases,  $(\s, \tau)$ can be considered as a feasible coreduction pair in the sequence $S$. Similarly, any simplex $\s$ added to $C$ during the processing of a lower star can be considered as a free simplex in the sequence $S$. So, Algorithm \ref{alg:coreduction} is a coreduction-based algorithm and then, thanks to Proposition \ref{prop2}, it returns a Forman gradient.
Moreover, since Algorithm \ref{alg:coreduction} pairs only simplices belonging to the same lower star and, by the definition of $F$, these simplices have the same filtration value. Thus, the returned Forman gradient $V$ is necessarily filtered with respect to $F$.
\end{proof}

\subsection{Extracting the discrete Morse complex}
\label{subsec:comp_boundary}

The discrete Morse complex $\M_*$ associated with a (filtered) Forman gradient $V$ on $\Ss$ is retrieved by navigating the paths of $V$.
The output consists of the boundary maps $\tde_k:\M_k\rightarrow \M_{k-1}$. These latter  can be seen as the arcs of a graph in which the nodes correspond to the critical simplices and each arc has a multiplicity which corresponds to a gradient path between two critical simplices.

Extracting the boundary maps by visiting the paths of $V$ may cause simplices to be visited more than once, as discussed in \cite{Robi11}.
In the worst case, a critical $k$-simplex may be connected by $V$-paths to all the $k$-simplices of $\Ss$ (this set is denoted as $\Sigma_k$). Moreover, each $k$-simplex of this set can be visited, via multiple $V$-paths, more than once; in the worst case each simplex will be visited $O(|\Sigma_k|)$ times.  The resulting worst-case complexity for retrieving the boundary maps of a single critical $k$-simplex can be quadratic in the number $|\Sigma_k|$ of $k$-simplices of $\Sigma$.

Even if this is a very rare case, some solutions have been proposed to guarantee lower complexity bounds by either using a Boolean function for marking the visited simplices \cite{gunther2011memory,Weis13}, or by using a priority queue \cite{Vijay12b} for limiting the number of simplices visited more than once.
Both approaches have limitations however. The approach in  \cite{gunther2011memory} is useful for reconstructing a combinatorial representation for the connectivity of the critical simplices, but it does not visit all the possibile paths, which is necessary for retrieving the correct boundary maps in $\mathbb{Z}$. The approach in \cite{Vijay12b} can successfully retrieve the correct boundary maps, but it requires a input scalar function to be defined all over the simplices of $\Sigma$.

The algorithm presented here is based on the general approach outlined in \cite{Robi11}.

% Algorithms for computing $V$-paths for constructing ascending and descending Morse complexes as in \cite{Felle14,Weis13, Vijay12b} in low dimensions cannot be extended since they do not compute all possible $V$-paths between critical simplices, which is necessary for defining the boundary maps.
% Thus, the  boundary maps extraction is performed by cycling over all the critical $k$-simplices in the Forman gradient $V$ and by visiting the gradient paths, according to a {\em descending traversal}, while searching for critical $(k-1)$-simplices.

% until reaching the set of critical $k-1$-simplices connected with $\tau$. All the $k$-simplices visited during the traversal are marked as visited. Then, for each critical $k-1$-simplex, an {\em ascending traversal} is performed by visiting all the $k$-simplices marked in the previous step.

% Given a critical $k$-simplex $\tau$, the computation of its boundary $\tde_k(\tau)$ with respect to $\M_*$ consists in two different steps.
% First, the gradient paths flowing from $\tau$ are computed through a {\em descending traversal} of the Forman gradient $V$. Then, for each critical $(k-1)$-simplex $\sigma$
% reached by a $V$-path starting from $\tau$, the gradient path is navigated in an {\em ascending traversal} in order to retrieve the multiplicity of $\sigma$ in the boundary with respect to $\M_*$ of $\tau$.

% \subsubsection{Descending traversal}

%

% !TEX root =  ../Pami2015.tex
\begin{algorithm}\caption{\textbf{- BoundaryMaps($\Sigma$,$\tau$,$V$)}}
    \label{alg:descending}
    \begin{algorithmic}[1]
        \medskip
        \STATE{INPUT: $\Sigma$, $d$-dimensional simplicial complex}
        \STATE{INPUT: $\tau$, critical $k$-simplex}
        \STATE{INPUT: $V$, Forman gradient}

        \STATE{OUTPUT: $M$, boundary maps as collections of arcs}

    		\STATE{$Q$ := $\emptyset$}
        \STATE{$Q.enqueue(\tau)$}

    		\WHILE{$Q \neq \emptyset$}
           \STATE{$\tau_0 := Q.dequeue()$}
	         \FOR{$\sigma_1 \in getBoundary(\tau_0 , \Sigma)$}
		           \IF{$isPaired(\sigma_1,V, \tau_1)$}
			              \STATE{$Q.enqueue(\tau_1)$}
                \ELSE
                    \STATE{$Add(M,\tau_1,\sigma_1)$}
                \ENDIF
		        \ENDFOR
	       \ENDWHILE
    \end{algorithmic}
\end{algorithm}

{\em Algorithm} \ref{alg:descending} illustrates the steps required for traversing the gradient paths in a descending fashion.
Starting from a critical $k$-simplex $\tau$, a breadth-first traversal is performed by navigating from $\tau$ to its adjacent $k$-simplices passing through their shared $(k-1)$-simplices. The breadth-first traversal is supported by a queue $Q$.
Given a $k$-simplex $\tau_0$ extracted from the queue $Q$ (row 8), we examine all the $(k-1)$-simplices $\s$ in the boundary of $\tau_0$ (row 11).
For each $(k-1)$-simplex $\s$, if $\s$ is paired with a $k$-simplex $\tau_1$ (row 12), $\tau_1$ is added to the queue (rows 15 and 16). If $\s$ is a critical simplex, then $\s$ is stored as on the boundary of $\tau$.

In Figure \ref{fig:traversal}(a), we show an example of the descending traversal performed by starting from critical triangle $\tau$.
For each edge on the boundary of $\tau$, the paired triangle is visited and enqueued (indicated in red in Figure \ref{fig:traversal}(b)). The process continues recursively for each new triangle (Figure \ref{fig:traversal}(c)) until the entire region associated with $\tau$ has been covered. When a critical edge $\s$ is encountered, the relation with $\tau$ is stored in the boundary maps (Figure \ref{fig:traversal}(d)).

\begin{figure*}
	\centering

	\begin{tabular}{cccc}
		\includegraphics[width=0.22\linewidth]{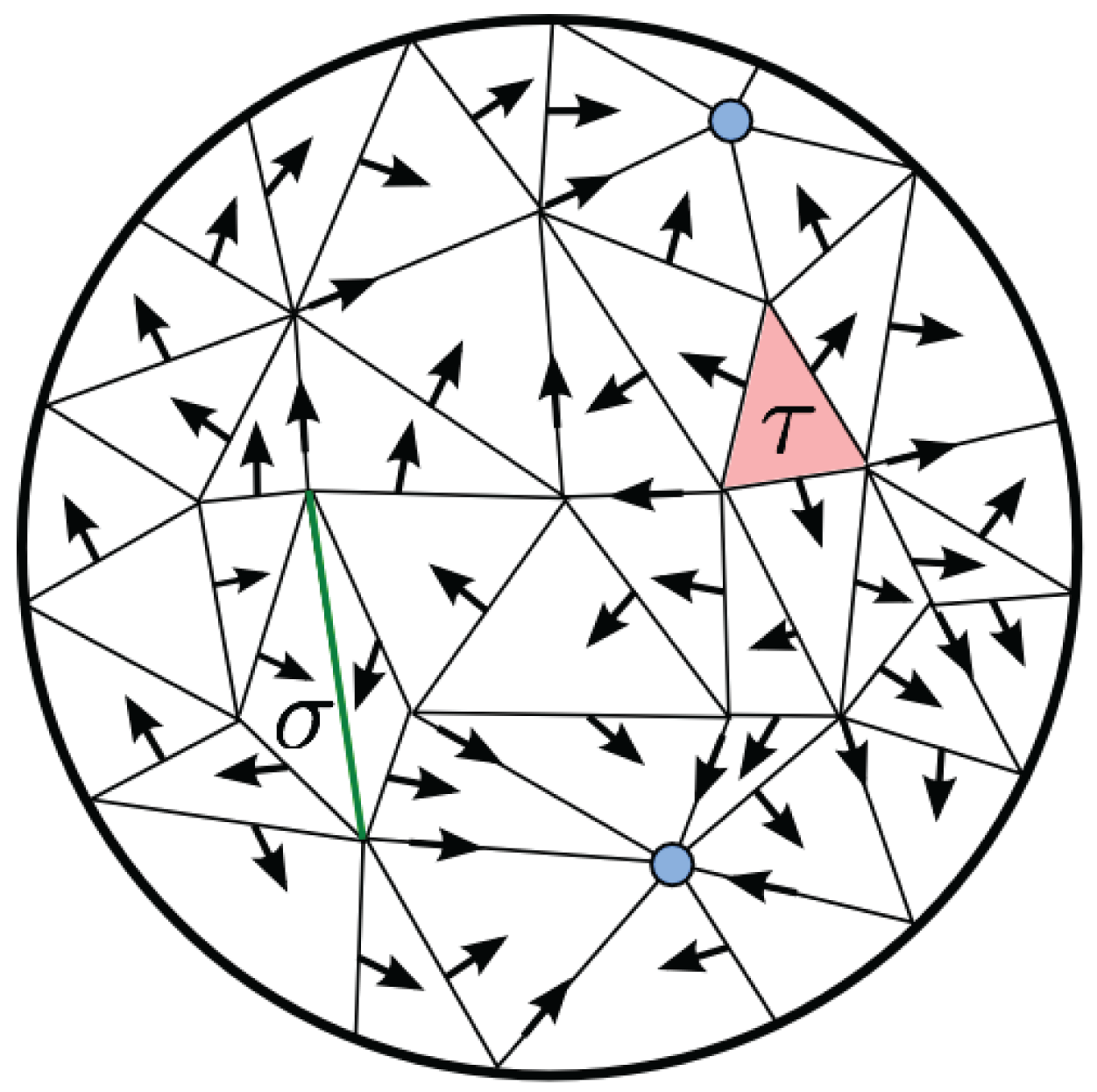} &
		\includegraphics[width=0.22\linewidth]{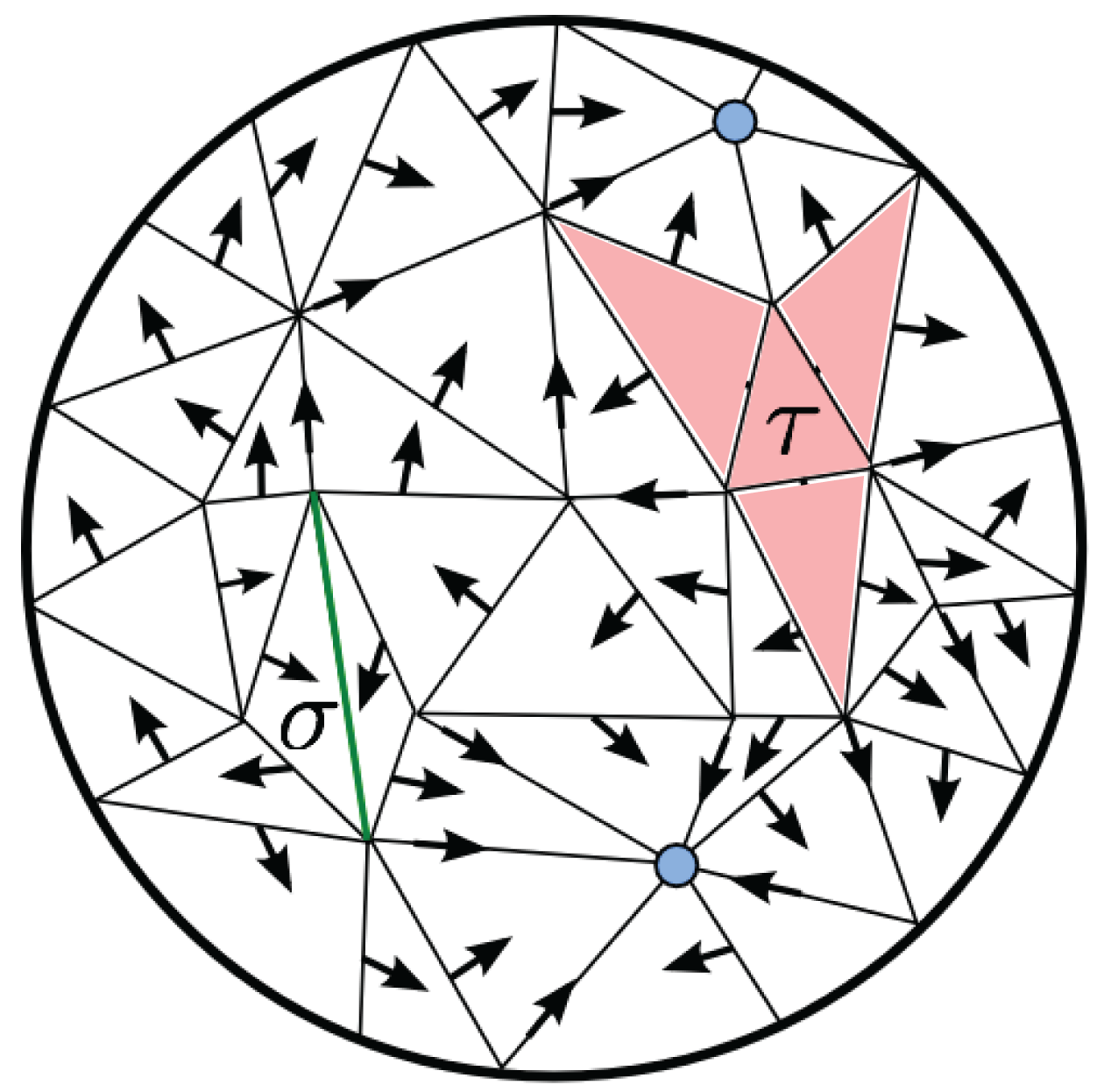} &
		\includegraphics[width=0.22\linewidth]{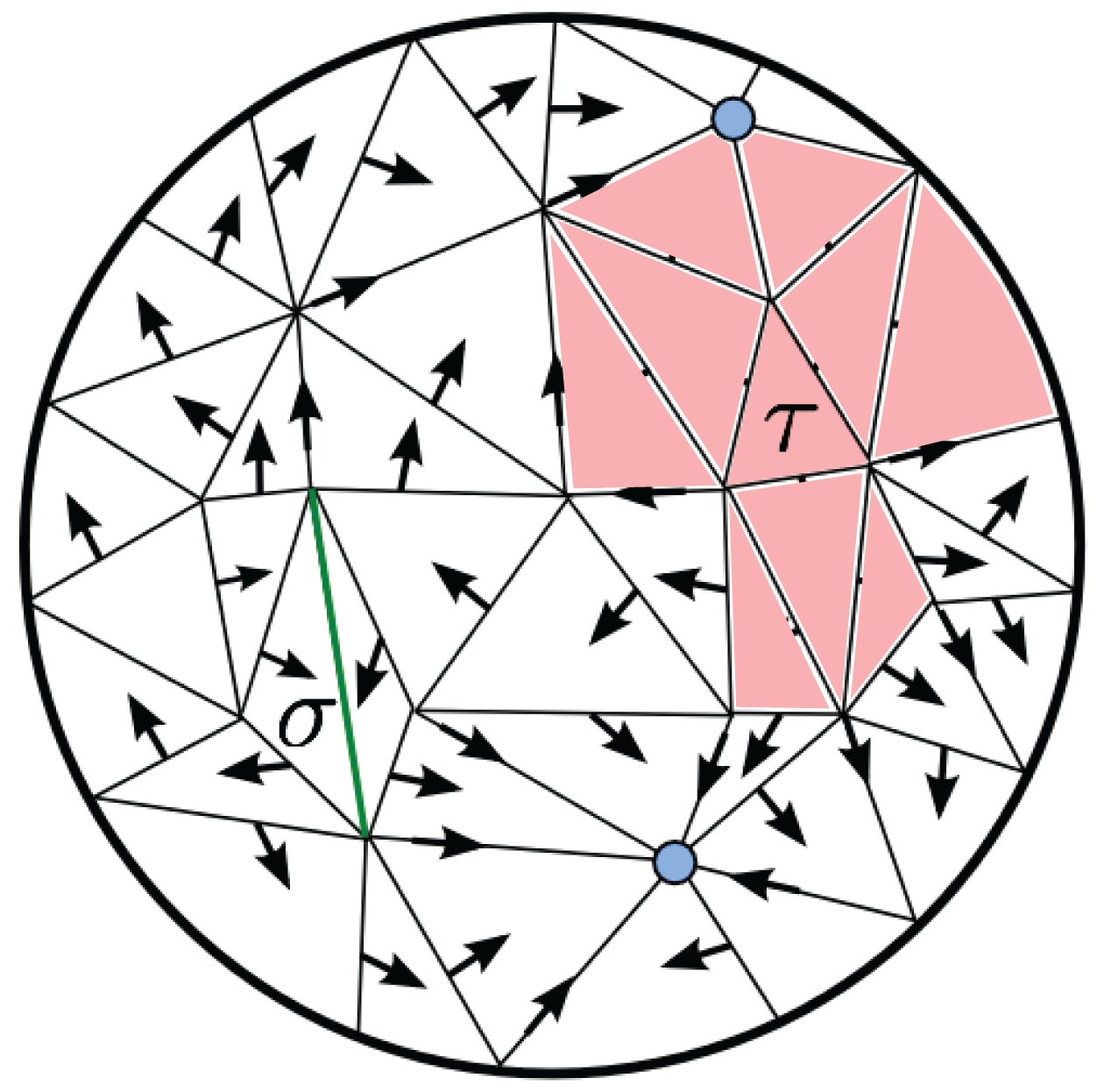} &
		\includegraphics[width=0.22\linewidth]{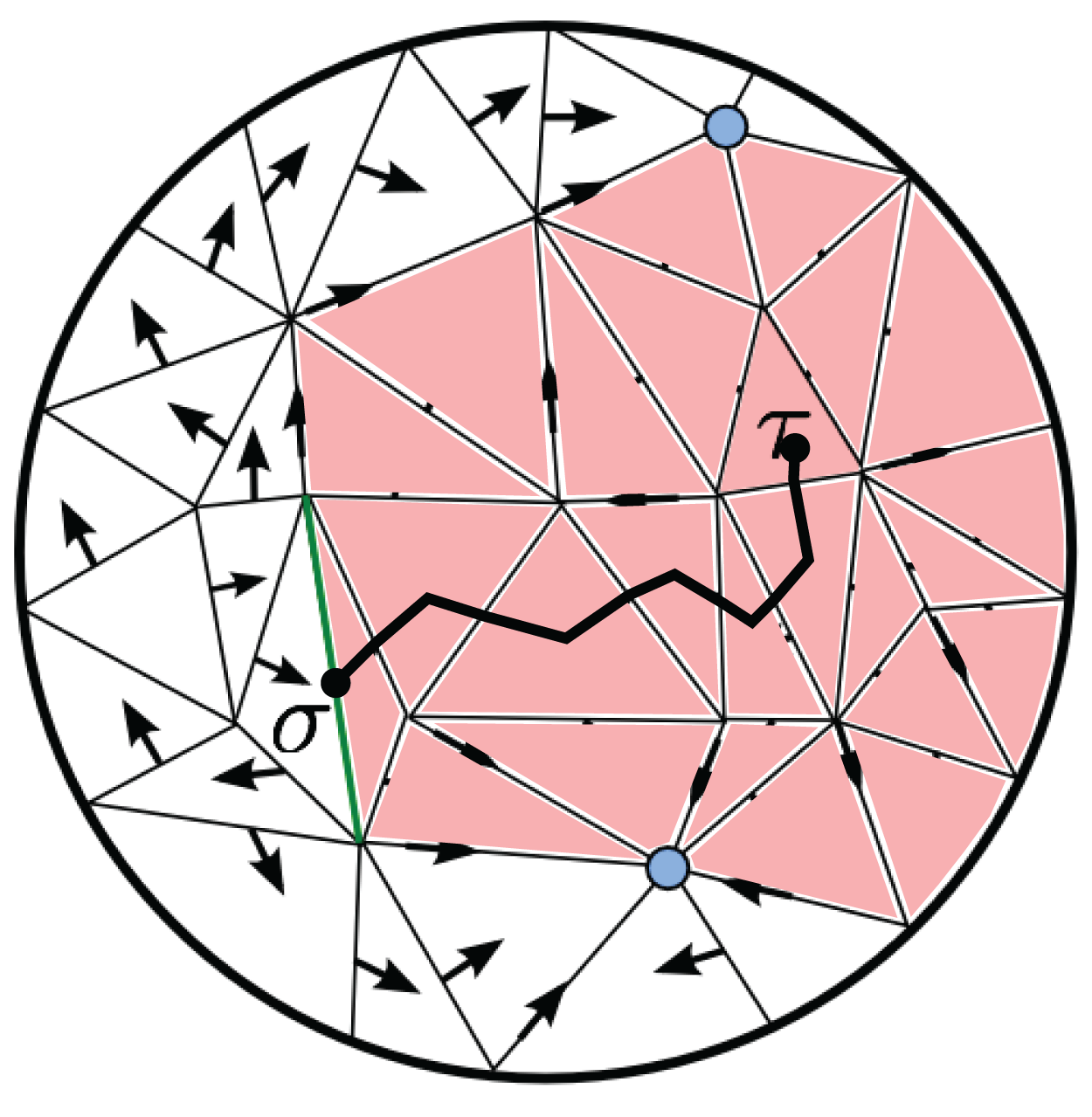} \\
		(a) & (b) & (c) & (d) \\
	\end{tabular}

	\caption{Descending traversal starting from $\tau$. Expanding the gradient $V$-paths the critical edge $\s$ is encountered and stored as connected to $\tau$.}
	\label{fig:traversal}
\end{figure*}

The procedures and the functions, on which Algorithm \ref{alg:descending} is based, are:

\begin{itemize}
    \item Function $getBoundary(\tau,\Sigma)$: returns the immediate boundary of $k$-simplex $\tau$, i.e., its $(k-1)$-faces. Extracting the immediate boundary is performed by taking all the combinations of the $k$ vertices of $\tau$, and it is a linear process in the number of vertices of $\tau$.

    \item Function $isPaired(\sigma,V,\tau_1)$: returns the value {\em True} if  $(k-1)$-simplex $\sigma$ is paired with a $k$-simplex in $V$ and the value {\em False} otherwise. In the former case it returns the paired simplex $\tau_1$. This is done by considering all the top simplices in the star of a vertex $w$ of $\sigma$ and visiting the gradient encoding of those top simplices which are incident in $\sigma$. The time complexity is $O(t_w)$  in the worst case, where $t_w$ is the number of top simplices incident in vertex $w$.

\end{itemize}

Algorithm \ref{alg:descending} is executed for each critical simplex in the Forman gradient. For each $k$-simplex $\tau$ popped from $Q$, the for loop is performed up to $k$ times. For each simplex $\sigma$ on the boundary of $\tau$ we check whether it is paired or not $O(t_w)$. Then, we can conclude that each iteration of the while loop takes  $O(kt_m)$, where $t_m$ is the maximum of the number of top simplices $t_k$ considered in $isPaired$ at the varies of $\sigma$. The algorithm has $O(qkt_m)$ worst-case time complexity, where $q$ is the number (counted with multiplicity) of $k$-simplices of $\Ss$ inserted in the queue $Q$.

\section{Experimental results}
\label{sec:experimental}

In this section, we  evaluate the performances of the coreduction-based  algorithm for Forman gradient computation and of the algorithm for computing the boundary maps that give a Morse complex, described in Section \ref{subsec:morse}, which are based on the encoding of the original simplicial complex as an $IA^*$ data structure. As described in Section \ref{sec:algorithm}, computing the Forman gradient focusing on the lower star of each vertex is an operation well suited for distributed, or parallel implementation. To test the gain in performances of such an approach, we have implemented also a parallel version of our gradient computation algorithm based on OpenMP. We compare our two implementations (sequential and parallel)  with the implementation provided by {\em Perseus} which computes the Morse complex using an $IG$ for encoding the input simplicial complex. To the extent of our knowledge, there are no implementations of the discrete Morse complex on a Simplex Tree.
%The Simplex Tree implemented in the {\em Gudhi} library \cite{gudhi2014} does not support the extraction of coboundary relations, which will make the implementation of the coreduction-based approach very inefficient.

In our experiments, we consider both real and synthetic datasets. The hardware configuration used is an Intel i7 3930K CPU at 3.20Ghz with 64GB of RAM. The data sets used in our experiments are described in Table \ref{table:storageStatic}. There are tetrahedralized volume data sets, and data sets obtained from networks and point clouds. Networks and point clouds  have no filtration provided as input.

\begin{table}
\centering
\resizebox{\columnwidth}{!}{%
\begin{tabular}{l | c c | c c c}

				Dataset & $|\Sigma|$ & $|C|$ & $IA^*$ & $IA_p^*$ & $IG$ \\ \hline

	\data{DTI-scan}       & 24M   & 0.14M  \tiny{(171x)} & 3.1m & 0.7m & 77.3h  \\
	\data{VisMale}         & 118M & 0.94M  \tiny{(125x)} & 29.2m & 6.5m & - \\
	\data{Ackley4}       & 204M & 0.01M \tiny{($10^4$x)} & 1.1h & 19.7m & -\\ \hline

	\data{Amazon1}     & 2.2M   & 0.16M \tiny{(13.7x)} & 14.5s & 3.7s & 20.9h \\
	\data{Amazon2}      & 18.4M & 0.37M \tiny{(49.7x)} & 281.9s & 68.3s & $>$200h \\
	\data{Roadnet}       & 4.8M   & 0.75M \tiny{(6.4x)} & 15.8s & 6.06s & $>$200h \\ \hline

	\data{S1.0}     & 0.6M  & 16 \tiny{$(10^5)$x} & 56.8s & 22.1s & 61.7s \\
	\data{S1.2}     & 26M   & 12 \tiny{$(10^7)$x} & 4.2h & 1.8h & -\\
	\data{S1.3}     & 197M & 7   \tiny{$(10^8)$x} & 173h & 74.3h & -\\

\end{tabular}%
}
\caption{Compression factor achieved by using the discrete Morse complex instead of the original simplicial complex. Column $|C|$ indicates the number of critical simplices, as opposed to the number of simplices $|\Sigma|$, for each dataset. Columns $IA^*$, $IA_p^*$ and $IG$ indicate the timings required for computing the discrete Morse complex with our sequential implementation, the multi-thread implementation, and the Perseus tool, respectively.}
\label{table:storageAndTime}
\end{table}

In Table \ref{table:storageAndTime}, we show first information about the size of the obtained discrete Morse complex (i.e., the number of cells), with respect to the original simplicial complex.
%column $|C|$ shows the number of critical simplices in the discrete gradient as well as the compression factor with respect to the simplicial complex taken as input.
The compression factor depends on the homological changes in the filtration of a dataset and on the dataset. Volumetric datasets benefit from a compression of about two orders of magnitude, network datasets are compressed by a factor of ten, while higher-dimensional complexes are compressed by five to eight orders of magnitude. This shows the advantage of using the Morse complex instead of the original one for computing homological information.

% We split the results obtained evaluating the storage cost (column {\em Space}) and the timings (column {\em Time}).

\begin{figure*}
	\centering
	\includegraphics[width=1.0\linewidth]{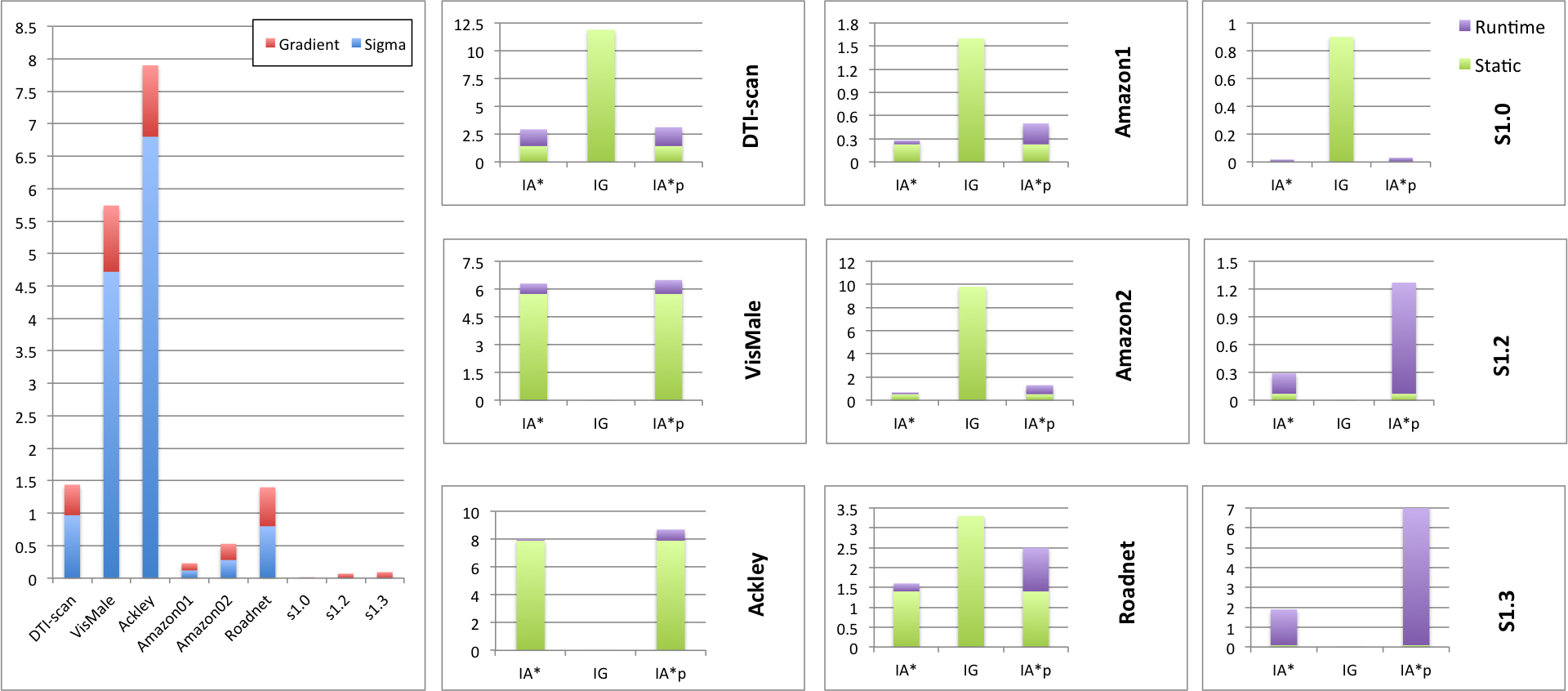}
	\caption{Storage cost required by computing and storing the Forman gradient and the discrete Morse complex.
	The first graph on the left indicates the amount of memory in GB required for storing the simplicial complex (blue bars) and the Forman gradient (red bars).
	The remaining graphs indicate, for each dataset, the amount of memory in GB used for storing the complex and the Forman gradient (green bars), and the overhead required at runtime for computing the gradient (purple bars). Results are presented comparing the $IA^*$ data structure, the $IG$,  and the parallel implementation based on the $IA^*$ data structure (indicated as $IA_p^*$). Missing columns represent experiments that exceeded the maximum amount of memory available.}
	\label{fig:graphSpace}
\end{figure*}

By comparing the timings, we see that our  approach (based on the $IA^*$ data structure) always outperforms Perseus (based on the  $IG$). When the number of simplices is low (dataset \data{Sphere-1.0}), the two implementations require a similar amount of time but, as soon as the number of simplices increases, our approach is  faster by two or three orders of magnitude. With the increasing of the dimension of the complex, we see that the complexity of computing the discrete Morse complex reaches its limits taking also 173 hours to complete for dataset \data{Sphere-1.3}. In our multi-threaded implementation, we have been able to use eight threads on our machine configuration, processing 8 vertices at a time. The speed up gained varies between a 2x and a 5x.

%As described in Section \ref{sec:algorithm}, computing the Forman gradient focusing on the lower star of each vertex is an operation well suited for distributed or parallel implementations. To test the gain in performances of such an approach, we have implemented a parallel version of our algorithm based on OpenMP. The results obtained are reported in Table \ref{table:storageAndTime} (column $IA_p^*$).

Figure \ref{fig:graphSpace} shows three evaluations. In the first graph, we are evaluating the memory used for representing the simplicial complex $\Ss$ (in blue) and the Forman gradient $V$ (in red).
We can notice that for the first three data sets, the complex is the entity requiring the highest amount of memory. For the remaining data sets, we notice that when the dimension increases, the storage cost decreases. For example, when comparing \data{Ackeley4} and \data{S1.0}, the total number of simplices is almost the same (see Table \ref{table:storageAndTime}, column $|\Sigma|$), while memory consumption is dramatically reduced, being dataset \data{S1.3} stored with less than 3.4MB compared to the 7.9GB required by \data{Ackeley4}. This is again due to the use of the IA$^*$ data structure  and to the encoding for the Forman gradient based on the top simplices.

% For the storage cost, in the case of the $IA^*$, we need to distinguish between the memory used for representing the simplicial complex $\S$ and the Forman gradient $V$ statically, and the memory used at runtime for computing both the Forman gradient and the boundary maps (column {\em $IA^*$ run.}).

While extracting the lower star in Algorithm \ref{alg:coreduction}, the $k$-simplices are recursively extracted from the $IA^*$ data structure and explicitly represented. This operation causes the main increase in the memory consumption at runtime. This is documented in the remaining graphs of Figure \ref{fig:graphSpace}. We are indicating in green the static overhead required for storing the simplicial complex and the Forman gradient and in purple the amount of memory used at runtime.

As we can notice (column $IA^*$), the difference between the static overhead and the dynamic overhead is larger when working on datasets in higher dimensions, while it becomes negligible when working in two or three dimensions. This fact is intrinsically related to the dimension $d$ of the original simplicial complex. When $d$ is small, the lower star of each vertex is also small. When working on higher dimensional complexes, the number of simplices in the lower star grows exponentially, since the number of simplices on the boundary of any $k$-simplex is exponential in $k$. In the worst-case scenario of our experiments (\data{S1.3}), the encoding of the star occupies 1.8GB at runtime, while storing the simplicial complex and the gradient requires less than 100MB.

The implementation in Perseus, based on the $IG$, does not present a difference between static and dynamic overhead, since all the simplices are already represented at the beginning and progressively simplified during the computation. Thus, the maximum peak is reached before starting the reduction algorithm. As a result, the $IG$ presents serious limitations when the dimension of the complex increases.

If we considering our parallel implementation (column $IA_p^*$), we see that the maximum peak of memory is higher, since all threads run on the same machine. Looking at the graphs in the first column (\data{DTI-scan}, \data{VisMale}, \data{Ackley4}), we recognize that the runtime overhead of this version is still comparable to the one of the single-thread implementation. This is an expected result since these are low dimensional data sets with a fairly small lower star for each vertex.
With the increasing in the data set dimension (second column), the overhead required by the parallel implementation starts to be relevant. In the worst case, we have experienced a memory overhead up to 6 times larger than the single thread implementation (third column data set \data{S1.3}). These results suggest that the whole framework is promising for a distributed environment,  where each process has its dedicated amount of memory.

\section{Concluding remarks}
\label{sec:conclusion}

We have studied different strategies to endow a simplicial complex with a Forman gradient through the use of homology-preserving operators and to extract the corresponding discrete Morse complex. We have formally proven the theoretical equivalence of such methods which allow for reducing the complexity of the computation through reductions and coreductions. We have developed and implemented  algorithms to efficiently build a discrete Morse complex based on coreductions, on a space-efficient representation of the simplicial complex and on a compact encoding of the Forman gradient, also implementing a parallel version of the latter.

Based on the results obtained from the parallel implementation, we are currently working on a distributed version of Algorithm \ref{alg:coreduction}. Since the process is localized within the star of each vertex, by distributing the computation on different machines, we expect to get a boost on timings without affecting memory consumption.

We are also considering the application of this work in single-parameter and multi-parameter persistent homology computation, as the basis for tools for shape understanding and retrieval,  and in segmentation of time-varying 3D scalar fields in the context of scientific data visualization.

% {\bf Persistent homology} \cite{edelsbrunner2008persistent} has been recognized as another fundamental tool for studying scientific data.
The best implementation currently available in the literature for computing {\em persistent homology} \cite{edelsbrunner2008persistent} on  high-dimensional complexes is based on annotations and on the Simplex Tree and it represents all the simplices of the simplicial complex explicitly \cite{boissonnat2013compressed}. This is also the case for any persistent homology computation algorithms based on boundary map reduction, because of the need to represent all the simplices explicitly.  This puts practical limitations when working on large  complexes. In these cases, our approach is particularly useful since the Morse complex is a simpler structure sharing the same persistent homology as the original simplicial complex.

%will have much smaller boundary matrices. as an intermediate representation when computing persistent homology. Our approach already provides the boundary maps of the Forman

% The representation of a Forman gradient will be at the base of many application domains.
%
% Specifically, the proposed algorithm can be used as a tool for enhancing the computation of extremum graphs and of multi-parameter persistent homology.\\

{\em Multi-parameter persistent homology} (also called multi-dimensional persistent homology) is  an extension of persistent homology  for data characterized by multiple parameters, like multi-field data sets. In this case, not a single filtration but multiple filtrations are considered. To date, the approaches proposed in the literature for computing multi-parameter persistent homology are at a pioneering level and are not able to deal with the complexity and the size of real datasets. Recently, an interesting connection between multi-parameter persistent homology and discrete Morse theory has been pointed out in \cite{landi13multi}. A  formal proof  is given of the equivalence between the multi-parameter persistent homology of the  Morse complex defined by a Forman gradient compatible with the multi-filtration and that of the underlying simplicial complex is provided. An algorithm has been proposed by Allili et al. \cite{Allili2017} for computing a Forman gradient on a vector-valued function, but its implementation is limited to triangle meshes of very small size.
%In the near future, we will consider the dimension-independent encoding described in this paper. B
Based on the dimension-independent encoding for the Forman gradient described in this paper, we are planning to develop a new algorithm  that  works independently of the dimension of the domain (the underlying simplicial complex) and of the codomain (the number of filtrations provided). A parallel
implementation will be also at the center of future studies for empowering the computation of  multi-parameter persistent homology.
%
%In this framework, the lower stars of the vertices no longer form a partition of the original complex $\Sigma$ and the parallelization obtained by using Algorithm \ref{alg:coreduction} cannot be extended. We have proposed in \cite{Iuricich2016} a new dimension-agnostic algorithm for computing a Forman gradient on multivariate data. The approach is a coreduction-based algorithm working on a suitable partition of the simplicial domain.
%
%A parallel implementation is also discussed for triangle meshes and cubical grids. The obtained discrete gradient shares the same properties of the Forman gradient described in this work (i.e., each simplex is paired with at most one other simplex). We are planning to use Algorithm \ref{alg:descending} for computing the descending paths between the critical cells found while studying the behaviour of the multifiltration.
%

In scientific visualization, {\em extremum graphs} have been defined as topological tools to understand and visualize the structure of 3D scalar fields, i.e., scalar fields defined at points in the three-dimensional Euclidean space \cite{spines}. The {\em extremum graph} is a subgraph of the graph representing  the boundary maps of the Morse complex. We recall that the boundary maps encode all the incidence relations between a critical $k$-simplex and a critical $(k-1)$-simplex, for $1\leq k \leq d=dim(\Ss)$. The extremum graph only represents the boundary maps between critical $d$-simplices and $(d-1)$-simplices and between critical 1-simplices and 0-simplices. For each pair of critical simplices,  it also encodes the chain of simplices that connects the two critical ones. In \cite{spines}, a visualization technique, called {\em topological spines},  has been developed specifically for extremum graphs not only of static 3D scalar fields, but also of time-varying fields, which can be regarded
as 4D scalar fields. The algorithms described in Section \ref{sec:algorithm} can be suitably adapted to efficiently compute the extremum graphs of a scalar field. Using the scalar function as a filtering function, we can compute the Forman gradient $V$ using Algorithm \ref{alg:coreduction}. The gradient paths of $V$ now describe the behavior of the input scalar field. Using Algorithm \ref{alg:descending}, we can extract the incidence relations between critical simplices by starting the descending traversal from critical $d$-simplices and from critical 1-simplices.
Our approach will make the computation of extremum graphs \cite{Narayanan2015} (and topological spines) feasible for 4D fields, but also for 3D fields defined on a tetrahedral mesh (as needed for complex 3D domains), while the current approach \cite{spines} works only works on scalar fields defined on cubic grids.

%Future developments are designing and implementing an efficient encoding for a simplicial complex in arbitrary dimensions based on the Stellar tree \cite{fellegara2017stellar}, a compact topological data structure based on a spatial index  for arbitrary simplicial complexes, which stores only the vertices and the top simplices of the complex.We expect the Stellar tree to be very efficient here since reductions and coreductions can be independently and massively applied to  portions of the complex inside the blocks of its spatial subdivision.
%This latter would not only reduce the storage cost further but also allows for an efficient localized homology computation.
%Based on the results obtained from the parallel implementation, we will consider implementing a distributed version of {\em Algorithm} \ref{alg:coreduction}. Subdividing the computation on different machines, we should be able to get a boost on timings without affecting memory consumption.

\section*{Acknowledgments}
This work has been partially supported by the US National Science Foundation under grant number IIS-1116747.
The authors wish to thank Davide Bolognini, Emanuela De Negri and Maria Evelina Rossi for their helpful comments and suggestions.

% \section*{References}
\bibliographystyle{elsarticle-num}
\bibliography{biblioGMODnoDoi}

\end{document}